\documentclass[conference]{IEEEtran}

\usepackage{setspace}
\usepackage{color}
\usepackage{graphicx}
\usepackage{url}
\usepackage{amssymb}
\usepackage{amsmath}
\usepackage{amsfonts}
\usepackage{subfig}
\usepackage[ruled]{algorithm2e}
\usepackage{balance}

\newtheorem{theorem}{Theorem}
\newtheorem{lemma}[theorem]{Lemma}

\newcommand{\prob}[1]{\mathbb{P}\left( #1 \right)}
\newcommand{\E}[1]{\mathbb{E}\left[ #1\right]}

\newcommand{\comment}[1]{}

\begin{document}

\title{Go-With-The-Winner: Client-Side Server Selection for Content Delivery}

\author{Chang Liu$^\dag$, Ramesh K. Sitaraman$^\dag$$^\ddag$, and Don Towsley$^\dag$\\
$^\dag$University of Massachusetts, Amherst\hspace{0.2in}$^\ddag$Akamai Technologies Inc.\\
\texttt{\{cliu, ramesh, towsley\}@cs.umass.edu}}


\maketitle

\begin{abstract}
Content delivery networks deliver much of the world's web and video content by deploying a large distributed network of servers. We model and analyze a simple paradigm for client-side server selection that is commonly used in practice  where each user independently measures the performance of a set of candidate servers and selects the one that performs the best. For web (resp., video) delivery, we propose and analyze a simple algorithm where each user randomly chooses two or more candidate servers and selects the server that provided the best hit rate (resp., bit rate). We prove that the algorithm converges quickly to an optimal state where all users receive the best hit rate (resp., bit rate), with high probability. We also show that if each user chose just one random server instead of two,  some users receive a hit rate (resp., bit rate) that tends to zero. We simulate our algorithm and evaluate its performance with varying choices of parameters, system load, and content popularity. 
\end{abstract}

\section{Introduction}
Modern content delivery networks (CDNs) host and deliver a large fraction of  the world's  web content, video content, and application services on behalf of enterprises that include most major web portals, media outlets, social networks, application providers, and news channels \cite{nygren2010akamai}.  CDNs deploy large numbers of servers around the world that can store content and deliver that content to users who request it. When a user requests  a content item, say a web page or a video,  the user is directed to one of  the CDN's servers that can serve the desired content to the user. The goal of a CDN is to maximize the performance as perceived by the user while efficiently managing its server resources.  

A key functionality of a CDN is the {\em server selection} process by which client software running on the user's computer or device, such as media player or a browser, is directed to a suitable server of a CDN \cite{dilley2002globally}.  The desired outcome of the server selection process is that each user is directed to a server that can provide the requested content with good performance. The metrics for performance that are optimized vary by the type of content being accessed. For instance, good performance for a user accessing a web page might mean that the web page downloads quickly. Good performance for a user watching a video might mean that the content is delivered by the server at a sufficiently high bit rate to avoid the video from freezing and rebuffering \cite{KrishnanS12}.

Server selection can be performed in two distinct ways that are not mutually exclusive. {\em Network-side server selection} algorithms monitor the real-time characteristics of the CDN and the Internet. Such algorithms are often complex and measure liveness and load of the CDN's servers, as well as latency, loss, and bandwidth of the communication paths between servers and users. Using this information, the algorithm computes a good ``mapping'' of users to servers, such that each user is assigned a ``proximal'' server capable of serving that user's content \cite{nygren2010akamai}. This mapping is computed periodically and is typically made available to client software running on the user's computer or device using the domain name system (DNS). Specifically, the user's browser or media player looks up the domain name of the content that it wants to download and receives as translation the ip address of the selected server.
 
A complementary approach to network-side server selection that is commonly is used is {\em client-side server selection} where the client software running on the user's computer or device embodies a server selection algorithm. The client software is typically unaware of the global state of the server infrastructure, the Internet, or other users. Rather, the client software typically makes future server selection decisions based on its own historical performance measurements from past server downloads. Client-side server selection can often be implemented as a plug-in within media players,  web browsers, and web download managers \cite{AkamaiDLM}.  
 
 While client-side server selection can be used to select servers within a single CDN, it can also be used in a multi-CDN setting. Large content providers often make the same content available to the user via multiple CDNs. In this case,  the client software running on the user's device tries out the different CDNs and chooses the ``best''  server  from across multiple  CDNs. For instance, NetFlix uses three different CDNs and the media player incorporates a client-side server selection algorithm to choose the ``best'' server (and the corresponding CDN) using performance metrics such as the video bit rates achievable from the various choices  \cite{adhikari2012unreeling}. Note also that in the typical multi-CDN  case, both network-side and client-side server selection are used together, where the former is used to choose the candidate servers from each CDN and the latter is used by the user to pick the ``best'' among all the candidates.

 \subsection{The Go-With-The-Winner paradigm}
A common and intuitive paradigm that is often used for client-side server selection in practice is what we call  {\em ``Go-With-The-Winner''}  that consist of an initial {\em trial period}  during which each user independently ``tries out'' a set of {\em candidate servers} by requesting content or services from them (cf. Figure~\ref{fig:serverselection}). Subsequently, each user independently {\em decides} on the ``best'' performing server using historical performance information that the user collected for the candidate servers during the trial period. It is commonly implemented in the content delivery context that incorporate choosing a web or video content server from among a cluster of such servers. 

Besides content delivery, the Go-With-The-Winner paradigm is also common for other Internet services, though we do not explicitly study such services in our work.  For instance, BIND, which is the most widely deployed DNS resolver (i.e., DNS client) on the Internet,  tracks performance as a smoothed value of the historical round trip times (called SRTT) from past queries  for a set of candidate name servers. Then BIND chooses a particular name server to query in part based on the computed SRTT values \cite{liu2009dns}. It is also notable that BIND implementations incorporate randomness in the candidate selection process.

\begin{figure}[t]
\centering
\includegraphics[width=0.35\textwidth]{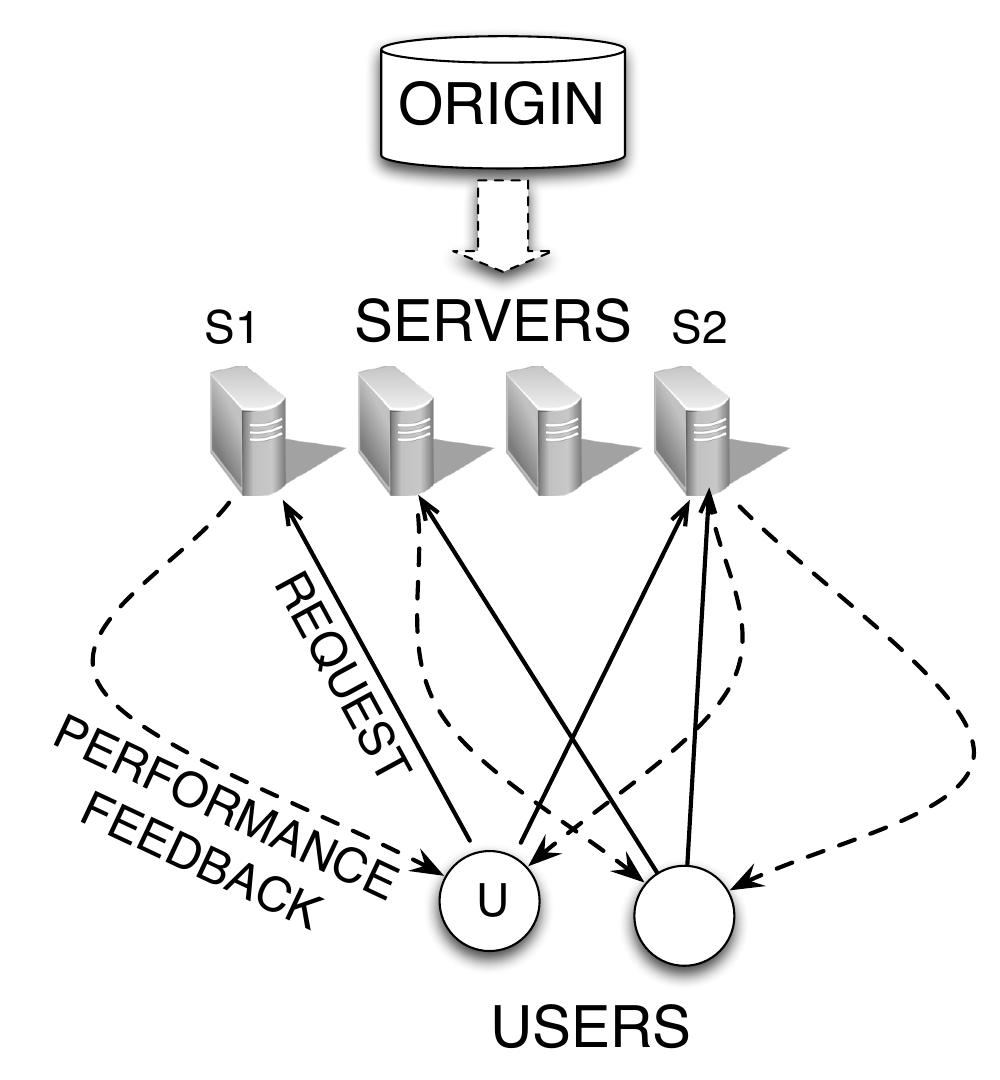}
\caption{Client-side Server Selection with the Go-With-The-Winner paradigm. User $U$ makes request to two candidate servers $S1$ and $S2$. After a trial period of observing the performance provided by the candidate, the user selects the better performing server.}
\label{fig:serverselection}
\end{figure}
The three key characteristics of the Go-With-The-Winner paradigm are as follows.
\begin{enumerate}
\item {\em Distributed control.} Each user makes decisions in a distributed fashion using only knowledge available to it. There is no explicit information about the global state of the servers or other users, beyond what the user can infer from it's own historical experience.
\item {\em Performance feedback only.} There is no explicit feedback from a server to a user who requested service beyond what can be inferred by the performance experienced by the user.
\item {\em Choosing the ``best'' performer.} The selection criteria is based on historical performance measured by the user and consists of selecting the best server according to some performance metric (i.e., go with the winner).
\end{enumerate}
Besides its inherent simplicity and naturalness, the paradigm is sometimes the only feasible and robust solution.  For instance, in many settings, the client software running on  the user's device that performs server selection has no detailed knowledge of the state of the server infrastructure as it is managed and owned by other business entities. In this case, the primary feedback mechanism for the client is its own historical performance measurements.

While client-side server selection is widely implemented, its theoretical foundations are not well understood. A goal of our work is to provide such a foundation in the context of web and video content delivery.  {\em It is not our intention to model a real-life client-side server selection process in its entirety which can involve other adhoc implemention-specific considerations. But rather we abstract an analytical model that we can explore to extract basic principles of the paradigm that are applicable in a broad context.}

\subsection{Our contributions}
 We propose a simple theoretical model for the study of client-side server selection algorithms that use the  Go-With-The-Winner paradigm. Using our model, we answer foundational questions such as how does randomness help in the trial period when selecting candidate servers? How many candidate servers should be selected in the trial phase? How long does it take for users to narrow down their choice and decide on a single server? Under what conditions does the selection algorithm converge to a state where all users have made the correct server choices, i.e., the selected servers provide good performance to their users? Some of our key results that help answer these questions follow.

{\em (1)}  In Section~\ref{sec:maxhitrate}, in the context of web content delivery, we analyze a simple algorithm called $\textrm{GoWithTheWinner}$ where each user independently selects two or more random servers as candidates and decides on the server that provided the best cache hit rate,. We show that with high probability, the algorithm converges quickly to a state where no cache is overloaded and all users obtain a 100\% hit rate. Further, we show that  two or more random choices of candidate servers are necessary, as just one random choice will result in some users (and some servers) incurring cache hit rates that tend to zero, as the number of users and servers tend to infinity. This work represents the first demonstration of the ``power of two choices''  phenomena in the context of client-side server selection for content delivery,  akin to similar phenomena observed in balls-into-bins games \cite{mitzenmacherRS2001}, load balancing,  circuit-switching algorithms \cite{cole1998randomized}, relay allocation for services like Skype \cite{Nguyen:2008}, and multi-path communication \cite{Peter:2007}. 

{\em (2)}   In Section~\ref{sec:maxbitrate}, in the context of video content delivery, we propose a simple algorithm called $\textrm{MaxBitRate}$ where each user independently selects two or more random servers as candidates and decides on the server that provided the best bit rate for the video stream,  We show that with high probability, the algorithm converges quickly to a state where no server is overloaded and all users obtain the required bit rate for their video to play without freezes. Further, we show that two or more random choices of candidate servers are necessary, as just one random choice will result in some users receiving bit rates that tend to zero, as the number of users and servers tends to infinity. 

{\em (3)} In Section~\ref{sec:empirical}, we go beyond our theoretical model and simulate algorithm $\textrm{GoWithTheWinner}$ in more complex settings. We establish an inverse relationship between the length of the history used for hit rate computation (denoted by  $\tau$) and the failure rate defined as the probability that the system converges to a non-optimal state. We show that as $\tau$ increases the convergence time increases, but the failure rate decreases. We  also empirically evaluate the impact of the number of choices of candidate servers. We show that two or more random choices are required for all users to receive a $100\%$ hit rate. Though even if only  70\% of the users  make two choices,  it is sufficient for $95\%$  of the users to receive a $100\%$ hit rate. Finally, we show that the convergence time increases with system load. But, convergence time decreases when the exponent of power law distribution that describes content popularity increases.

\section{Hit Rate Maximization for Web Content}
\label{sec:maxhitrate}
The key measure of web performance is {\em download time} which is the  time taken for a user to download a web object, such as a html page or an embedded image. CDNs enhance web performance by deploying a large number servers in access networks that are ``close'' to the users. Each server has a cache that is capable of storing web objects. When a user requests an object, such as a web page,  the user is directed to a server that can serve the object (cf. Figure~\ref{fig:serverselection}). If the server already has the object in its cache, i.e, the user's request is a  {\em cache hit}, the object is served from the cache to the user. In this case, the user experiences good performance, since the CDN's servers are proximal to the user and the object is downloaded quickly.  However, if the requested object is not in the server's cache, i.e., the user's request is a {\em cache miss}, then the server first fetches it from the origin, places it in its cache, and then serves the object to the user. In the case of a cache miss, the performance experienced by the user is often poor since the origin server is typically far away from the server and the user. In fact, if there is a cache miss, the user would have been better off not using the CDN at all, since downloading the content directly from the content provider's origin would likely have been faster!  Since the size of a server's cache is bounded, cache misses are inevitable.  A key goal of server selection for web content delivery is to jointly orchestrate server assignment and content placement in caches such that the cache hit rate is maximized. While server selection in CDNs is a complex process \cite{nygren2010akamai}, we analytically model the key elements that relate to content placement and cache hit rates,  leaving other factors that impact performance such as server-to-user latency for future work.

\subsection{Problem Formulation}
Let  $U$ be a set of $n_u$ users who each request an object picked independently from a set $C$ of size $n_c$ using a power law distribution where the $k^{th}$ most popular object in $C$ is picked with a probability \begin{equation}
p_k \stackrel{\Delta}{=} \frac{1}{k^\alpha \cdot \mathcal{H}(n_c, \alpha)}, \label{eq:powerlaw}
\end{equation} where $\alpha \geq 0 $ is the exponent of the distribution and $H(n_c, \alpha)$ is the generalized harmonic number that is the normalizing constant, i.e.,  $\mathcal{H}(n_c, \alpha) = \sum_{k=1}^{n_c} 1/k^{\alpha} $. Note that power law distributions (aka Zipf distributions) are commonly used to model the popularity of  online content such as web pages, and videos. This family of distributions is parametrized by a Zipf rank exponent  $\alpha$ with $\alpha = 0$ representing the extreme case of an uniform distribution and larger values of $\alpha$ representing a greater skew in the popularity. It has been estimated that the popularity of web content can be modeled by a power law distribution with an $\alpha$  in the range from 0.65 to 0.85 \cite{Breslau:1999,Gill:2007,Fricker:2012}.
The user then sticks with that content and makes a sequence of requests to the set of available servers. Relating to the reality, users tend to stay with one website for a while, say reading the news or looking at a friend's posts. Here the \textsl{whole website} is what we considered a content. We model  the sequence of requests generated by each user as a Poisson process with a homogeneous arrival rate $\lambda$.  Note that each request from user $u$ can be sent to one or more servers selected from $S_u \subseteq S$, where $S_u$ is the server set chosen by user $u$.

Let  $S$ be the set of $n_s$ servers that are capable of serving content to the users. Each server can cache at most $\kappa$ objects and a cache replacement policy such as LRU is used to evict objects when the cache is full. Given that the download time of a web object is significantly different when the request is a cache hit versus a cache miss, we make the reasonable assumption that the user can reliably infer if its request to download an object from a server resulted in a  cache hit or a cache miss immediately after the download completes.

The objective of client-side server selection is for each user $u \in U$ to independently  select a server $s \in S$ using only the performance feedback obtained on whether each request was a hit or a miss. Let  the hit rate function $H(u,s,t)$ denote the probability of user $u$ receiving a hit from server $s \in S_u$ at time $t$. We define the system-wide performance measure $H(t)$, as the best hit rate obtained by the worst user at time $t$,
 \begin{equation}\label{eq:Ht}
  H(t) \stackrel{\Delta}{=} \min_{u \in U} \max_{s \in S_u} H(u, s, t),
\end{equation}
a.k.a. the \textsl{minmax hit rate}.  Our goal is to maximize $H(t)$.In the rest of the section, we describe a simple canonical  ``Go-With-The-Winner''  algorithm for server selection and show that it converges quickly to an optimal state, with high probability.

{\em Note:} Our formulation is intentionally simple so that it could model a variety of other situations in web content delivery. For instance, a single server could in fact model a cluster of front-end servers that share a single backend object cache. A single object could in fact model a bucket of objects that cached together as is often done in a CDN context \cite{nygren2010akamai}. 

\subsection{The \textsl{GoWithTheWinner} Algorithm}
\comment{After each user $u \in U$ has picked an content item using the power law distribution described in Equation~\ref{eq:powerlaw},  algorithm \textsl{GoWithTheWinner} described below is executed independently by each user $u \in U$ to select a server that's likely to have the content. In this algorithm, each user locally executes a simple ``Go-With-The-Winner'' strategy of trying out $\sigma$ randomly chosen candidate servers initially.  Then, using the past hit rate over a time window of length $\tau$ as feedback, each user independently either chooses to continue with all the servers in $S_u$ or decides on a  single server that provided the best performance.  If multiple servers provided a $100\%$ hit rate in line 8 of the algorithm, the user decides to use the first one found.} 

After each user $u \in U$ selects a content item and a set of $\sigma$ servers $S_u$,  the user executes algorithm \textsl{GoWithTheWinner} to select a server likely to have the content. In this algorithm, each user locally executes a simple ``Go-With-The-Winner'' strategy of trying out $\sigma$ randomly chosen candidate servers initially. For each server $s\in S_u$, the user keeps track of the most recent request results in a vector $\mathbf{h}^s=(h^s_1,h^s_2,\cdots,h^s_{\tau})$ where $h^s_k=1$ is the $k$-th recent request results in a hit from server $s$ and $h^s_k=0$ if otherwise. We call $\tau$ the sliding window size. Using the hit rates, each user then independently either chooses to continue with all the servers in $S_u$ or decides on a  single server that provided good performance. If there are multiple servers providing $100\%$ hit rate, the user decides to use the first one found.


  \LinesNumbered
  \PrintSemicolon
  \SetAlgoLined
  \SetNlSty{textsf}{}{}
  \begin{algorithm}[t]
    \caption{GoWithTheWinner}
    Each user $u$ independently chooses a random subset $S_u \subseteq S$ of candidate servers such that $|S_u| = \sigma$ and does the following.\\
    \For{each $s\in S_u$}{
      set $\mathbf{h^s} \leftarrow (h^s_1,h^s_2,\cdots,h^s_{\tau})=\mathbf{0}$\;
      }
      \For{each arrival of request}{
        set $t$ to the current time\;
        Request content $a_u$ from {\em all} servers $s \in S_u$\;
        \For{each server  $s \in S_u$}{
        		$h^s_i \leftarrow h^s_{i-1}, 2\leq i\leq\tau$\;
		$h^s_1 \leftarrow \text{if hit}; h^s_1 \leftarrow 0, \text{if miss}$\;
        		compute hit rate $H_\tau(u,s, t) \leftarrow (\sum_{i=1}^{\tau}h^s_i)/\tau$ \;
         	\If{$H_\tau(u,s, t) = 100\%$}{
        			{\em decide} on server $s$ by setting $S_u \leftarrow \{s\}$\;
			return\;
        		}
	}
    }
  \end{algorithm}

\subsection{Analysis of Algorithm MaxHitRate}
Here we rigorously analyze the case where $n_u = n_c = n_s = n$ and experimentally explore other variants where $n_c$ and $n_u$ are larger than $n_s$ in Section~\ref{sec:nu>ns} and~\ref{sec:empirical}.  Let $H(t)$   be as defined in (\ref{eq:Ht}). If $\sigma \geq 2$, we show that with high probability $H(t) = 100\%$, for all $t \geq T$, where $T = O(\frac{\kappa}{\log (\kappa + 1)} (\log n)^{\kappa+1}\log\log n)$. That is,  the algorithm converges quickly with high probability to an optimal state where {\em every} user has decided on a single server that provides a 100\% hit rate,  and {\em every} server has the content requested by its users.

{\em Definitions.} A server $s$  is said to be {\em overbooked} at some time $t$ if users request more than $\kappa$ distinct content items from server $s$, where $\kappa$ is the number of content items a server can hold. Note that a server may have more than $\kappa$ users and not be overbooked, provided the users collectively request a set of $\kappa$ or fewer content items. Also, note that a server that is overbooked at time $t$ is overbooked at every $t' \leq t$ since the number of users requesting  a server can only remain the same or decrease with time. Finally, a user $u$ is said to be {\em undecided} at time  $t$ if   $|S_u| > 1$ and is said to be {\em decided} if it has settled on a single server to serve its content and $|S_u| = 1$. Note that each user starts out undecided at  time zero, then decides on a server at some time  $t$ and remains decided in all future time later than $t$. Users calculate the hit rates of each of the available servers based on a history record of the last $\tau$ requests, where $\tau$ is called the sliding window size. 

\begin{lemma}\label{lem:overbookhit}
If the sliding window size $\tau = \Theta(\log^{\kappa+1}n)$, the probability that some user $u \in U$ decides on an overbooked server $s \in S_u$ upon any request arrival is at most  $1/n^{\Omega(1)}$. 
\end{lemma}
\begin{proof}
If user $u$ decides on server $s$  then the current request together with the previous $\tau-1$ requests are all hits. Let $H_k$, $k=1,2,\cdots,\tau$ be Bernoulli random variables, s.t. $H_k = 1$ if the most recent $k$-th request of $u$ is a hit and $H_k=0$ if it is a miss. To prove Lemma \ref{lem:overbookhit} we need to show 
\begin{equation}\label{eq:hittau}
\prob{\cap_{k=1}^{\tau}(H_k=1)} \leq n^{-\Omega(1)}.
\end{equation}
Let $t_0$ be the time a request for content $a$ from $u$ is generated and appears at server $s$. Let $t_0-\Delta$ be the time that the last request for $a$ arrives at $s$. Let $H$ be an indicator variable so that $H=1$ if the request at $t_0$ resulted in a hit and $H=0$ if resulted in a miss. Let $A_s=\{a_1,a_2,\cdots,a_M\}$ be the set of different content items requested at $s$, where $M>\kappa$. Let $N_i$ be the number of users requesting $a_i$ from $s$. WLOG, let $a_1=a$ be the content that $u$ requests. $\Delta$ is an exponential random variable, $\Delta\sim Exp(N\lambda)$, where $N=N_1$ is the number of users requesting $a$ at server $s$. Let $X_i, i=2,3,\cdots,M$ be an indicator that a request for $a_i$ arrives at $s$ during time interval $(t_0-\Delta,t_0), X_i\sim Bernoulli(1-e^{-N_i\lambda\Delta})$.
Thus, random variabe $Y=\sum_{i=2}^M X_i$ is the number of requests for different content items that arrive in the time interval. With the server running LRU replacement policy, 
\begin{equation}\label{eq:hitrate2arrivals}
\prob{H=0} = \prob{Y\geq \kappa},
\end{equation}
because more than $\kappa$ different requests other than $a$ must have arrived for content $a$ to be swapped out of the server. (\ref{eq:hitrate2arrivals}) shows that $H$ only depends on the arrival of other requests, which means events $H_k, k=1,2,\cdots,\tau$ are mutually independent. Furthermore\footnote{random variables $U\geq_dV$ if $\prob{U>x}\geq\prob{V>x}$ for all $x$.},
$$
Y = \sum_{i=2}^M X_i \geq_{d} \sum_{i=2}^M X',
$$
where $X'\sim Bernoulli(1-e^{-\lambda\Delta})$. Furthermore because $M\geq\kappa+1$, 
$$
Y \geq_{d} Z,
$$
where $Z\sim Binomial(\kappa, (1-e^{-\lambda\Delta}))$.

Thus, we have 
\begin{align*}
\prob{Y\geq\kappa} & \geq \prob{Z\geq \kappa}\\
	& = \int_0^\infty \prob{Z\geq\kappa|\Delta=t}f_{\Delta}(t) dt \\
	& = \int_0^\infty (1-e^{-\lambda t})^{\kappa}N\lambda e^{-N\lambda t} dt \\
	& = \frac{N!\kappa !}{(N+\kappa) !} \\
	& \geq (N+\kappa)^{-\kappa},
\end{align*}
where $f_{\Delta}(t)$ is the probability density function of $\Delta$.

Note that $N$ is the number of users requesting $a$ at server $s$, and is bounded by $N = O(\frac{\log n}{\log\log n})$, with high probability \cite{raab1998balls}.

Now, we can finally prove (\ref{eq:hittau}). Let $c'$ be an appropriate constant,\\
\begin{align*}
\prob{\cap_{k=1}^{\tau}(H_k=1)} & = \prob{H=1}^{\tau} = (1-\prob{H=0})^{\tau} \\
	& = (1-\prob{Y\geq\kappa)}^{\tau} \\
	& \leq (1-(N+\kappa)^{-\kappa})^{\tau} \\
	& \leq (1-(c'\frac{\log n}{\log\log n}+\kappa)^{-\kappa})^{\tau},
\end{align*}
which is $n^{-\Omega(1)}$ when $\tau = \Theta(\log^{\kappa+1}n)$.
\end{proof}

By bounding the time for $\tau$ requests to arrive at user $u$, we have the following,
\begin{lemma}\label{lem:timefortau}
If user $u$ is not \textsl{decided} with server $s\in S_u$ at time $t$, then the server is \textsl{overbooked} at time $t-\delta$ for $\delta = \frac{\tau+1}{\lambda}c_{0}$ where $c_0>1$ is a constant, with high probability.
\end{lemma}

\begin{proof}
Let random variable $N_{\delta}$ be the number of requests from $u$ during time $(t-\delta, t), N_{\delta}\sim Poisson(\lambda\delta)$.
A bound on the tail probability of Poisson random variables is developed in \cite{PoissonTailBound} as
$$
	\prob{X\leq x} \leq \frac{e^{-\lambda'}(e\lambda')^x}{x^x},
$$
where $X\sim Poisson(\lambda')$ and $x<\lambda'$.

Based on that we can show there are at lease $\tau+1$ requests during $(t-\delta,t)$ w.h.p. as the following,
\begin{align*}
\prob{N_{\delta} < \tau+1} & \leq e^{-\lambda\delta}\frac{(e\lambda\delta)^{\tau+1}}{(\tau+1)^{\tau+1}} = e^{-(\tau+1) c_{0}}(ec_0)^{(\tau+1)}\\
	& = e^{-(\tau+1)(c_0-1)}c_0^{(\tau+1)}\\
	& = n^{-\frac{(\tau+1)}{\log n}(c_0 - 1 - \log c_0)}\\
	& = n^{-\Theta(\log^{\kappa}n)},
\end{align*}
as $c_0>1$ and $\tau= \Theta(\log^{\kappa+1}n)$. Thus, w.h.p. no less than $\tau+1$ requests arrives at $u$. And because user is not decided at time $t$ we know that with high probability, at least one of previous $\tau$ requests receives a miss, which mean between the previous $(\tau+1)$-th request and the miss, there are $\kappa$ different other requests arrive at the server. Thus server $s$ is \textsl{overbooked} at the time the previous $(\tau+1)$-th request arrives, which with high probability is no earlier than $t-\delta$.
\end{proof}

Based on Lemmas \ref{lem:overbookhit} and \ref{lem:timefortau}, we can then establish the following theorem about the performance of Algorithm \textsl{GoWithTheWinner}.

\begin{theorem}
\label{thm:hitrate}
With probability at least $1 - \frac{1}{n^{\Omega(1)}}$,  the minmax hit rate $H(t) = 100\%$ for all $t \geq T$, provided $\sigma \geq 2$ and 
$$T = O(\frac{\kappa}{\log (\kappa + 1)} (\log n)^{\kappa+1} \log\log n).$$
 That is,  with high probability, 
 algorithm \textsl{GoWithTheWinner} converges by time $T$ to an optimal state where each user $u \in U$ has decided on a server $s \in S$ that serves it content with a $100\%$ hit rate. 
\end{theorem}

This is the main result for the performance analysis of the algorithm. Due to space limit, please see appendix for detailed proof of this theorem. 

Are two or more random choices necessary for all users to receive a $100\%$ hit rate? Analogous to the ``power of two choices''  in the balls-into-bins context \cite{mitzenmacherRS2001}, we show that two or more choices are required for good performance. 
\begin{theorem}
\label{thm:onechoice}
For any fixed constants $0 \leq \alpha < 1$ and $\kappa \geq 1$, when algorithm $\textrm{MaxHitRate}$ uses one random choice for each user ($\sigma = 1$), the minmax hit rate $H(t) = o(1),$ with high probability, i.e.,  $H(t)$ tends to zero as $n$ tends to infinity, with high probability.  
\end{theorem}
Please see appendix for the proof.
\subsection{The case when $n_u\gg n_s$}\label{sec:nu>ns}
Now we analyze the case that there are much more users than the number of servers, which is often the case in reality. 
Let $Y_i$ be the number of users associated with $i$ and $Y=\text{max}_{i\in S} Y_i$ be the maximum over all servers.  
Assuming $\sigma=1$ so that all users are initiated with only one randomly selected server, we have the following results on the maximum incoming users over all servers $Y$ and server capacity $\kappa$ for the system to converge to the optimal state that every user gets hit rate $1$.
\begin{theorem}
\label{lem:nu>ns}
\begin{enumerate}
\item When $n_s=n, n_u = n\log n$, with probability $1 - \frac{1}{n_s^{\Theta(1)}}$,  the maximum load (number of associated users) over all servers $Y\leq(1+\sqrt{3})\frac{n_u}{n_s} =(1+\sqrt{3})\log n$. If $\kappa\geq (1+\sqrt{3})\log n$, all users have hit rate 1.\\
\item When $n_s=n, n_u = n^{\alpha}, \alpha>1$, with probability $1 - \frac{1}{n_s^{\omega(1)}}$, the maximum load over all servers $Y=\frac{n_u}{n_s}=n^{\alpha-1}$. Thus if $\kappa\geq \frac{n_u}{n_s} = n^{\alpha-1}$, all users get hit rate $1$.
\end{enumerate}
\end{theorem}

Theorem~\ref{lem:nu>ns} implies that when $n_u>> n_s$ all the servers have balanced load of $\frac{n_u}{n_s}$, thus we don't need more server selection mechanism for load balancing other than just letting all users randomly choose the server. And in this case, it's not beneficial to let users start with more than $1$ randomly selected servers, because with $\sigma=1$ the load on all servers are balanced already. Thus, as long as we have feasible server capacity $\kappa\geq(1+\sqrt{3})\frac{n_u}{n_s}$ for $n_u=n_s\log n_s$ and $\kappa\geq\frac{n_u}{n_s}$ for $n_u=n_s^{\alpha}, \alpha>1$, all the users will have enough resources from the server and have $100\%$ hit rate by randomly select $1$ server.

The number of content items $n_c$ here doesn't not affect the result of load balancing. Actually, the result stays the same when $n_c\geq n_u$. And when the number of content items is much smaller than number of users, $n_c<<n_u$, the cache size can become smaller ($\kappa < \frac{n_u}{n_s}$) because the number of distinct requests at each server becomes smaller.

\section{Bit rate Maximization for Video Content}
\label{sec:maxbitrate}
In video streaming, a key performance metric is the {\em bit rate} at which an user can download the video stream. If the server is unable to provide the required bit rate  to the user, the video will freeze frequently resulting in an inferior viewing experience and reduced user engagement \cite{KrishnanS12}.  For simplicity, we model the server's bandwidth capacity that is often the critical bottleneck  resource, while leaving other factors that could  influence video performance such as the server-to-user connection and the server's cache\footnote{Unlike the web, cache hit rate is a less critical determinant of  video performance. Videos are cached in chunks by the server. The next chunk is often prefetched from origin if it is not  in cache, even while the current chunk is being played by the user,  so as to hide the origin-to-server latency.} for future work. 

\subsection{Problem formulation}
The bit rate required to play a stream without freezes is often the encoded  bit rate of the stream. For simplicity, we assume that each user requires a bit rate of 1 unit for playing its video and each server has the capacity to serve $\kappa$ units in aggregate. And we assume each server evenly divides its available bit rate capacity among all users who keeps a streaming connection with it. We make the reasonable assumption that each user can compute the bit rate that it receives from its chosen candidate servers and that this bit rate is used as the performance feedback (cf. Figure~\ref{fig:serverselection}).

Different from the delivering web content, where users make repetitive requests to the same website with Poisson processes, we consider users for video streaming have persistent connection with the server. We use a discrete time model in this case as compared to web content delivery where everything is in continuous time. We assume after each time unit, the users look at the bit rate provided by each of the available servers and then make decisions according to the performance (measured by bit rate).
The goal of each user is to find a server who can provide the required bit rate of 1 unit for viewing the video.

\subsection{Algorithm MaxBitRate}
After each user $u \in U$ has picked a video object $c_u \in C$ using the power law distribution described in Equation~\ref{eq:powerlaw}, Algorithm $\textrm{MaxBitRate}$ described below is executed independently by each user $u \in U$, in discrete time steps.
\begin{enumerate}
\item Choose a random subset of candidate servers $S_u \subseteq S$  such that $|S_u| = \sigma$.
\item At each time step $t \geq 0 $, do the following:
\begin{enumerate}
\item Request the video content from all servers $s \in S_u$.
\item For each server  $s \in S_u$, compute  $B(u,s, t) \stackrel{\Delta}{=}$ bit rate provided by server $s$ to user $u$ in the current time step.
\item If there exists a server $s \in S_u$ such that $B(u,s, t) = 1$, then {\em decide} on server $s$ by setting $S_u \leftarrow \{s\}$.
\end{enumerate}
\end{enumerate}
Note that the users are executing  a simple strategy of trying out $\sigma$ randomly chosen servers initially.  Then, using the bit rate received in the current time step as feedback, each user independently narrows it's choice of servers to a single server that provided the required unit bit rate. 
If multiple servers provided the required bit rate,  the user decides to use an arbitrary one. Further, note that a user $u$  downloading from a server $s$ at time $t$  knows  immediately whether or not the server is overloaded, since server $s$ is overloaded iff user $u$ received a bit rate of less than 1 unit from the server, i.e., $B(u,s,t) < 1$. This is a point of simplification in relation to the more complex situation for hit rate maximization where any single cache hit is not indicative of a non-overloaded server and a historical average of hit rates over a large enough time window $\tau$ is required as a probabilistic indicator of server overload. And furthermore, this simplification yields both faster convergence to an optimal state in $T = O(\log\log n/\log (\kappa + 1))$ steps and a much simpler proof of that convergence.

\subsection{Analysis of Algorithm MaxBitRate}
As before, we rigorously analyze the case where $n_u = n_s = n_c = n$. Let the {\em minmax} bit rate $B(t)$ be the best bit rate obtained by the worst user at time $t$, i.e., 
$$B(t) \stackrel{\Delta}{=} \min_{u \in U} \max_{s \in S} B(u, s, t).$$
\begin{theorem}
\label{thm:bitrate}
When $\sigma \geq 2$, the minmax bit rate converges to $B(t) = 1$ unit, for all $t \geq T$,  within time $T = O(\log\log n/\log (\kappa+1))$, with high probability. When $\sigma = 1$ on the other hand, the minmax bit rate $B(t) =O(\kappa \log\log n/\log n),$ with high probability. In particular, when $\sigma = 1$ and the cache size $\kappa$ is $o(\log n / \log\log n)$, including the case when $\kappa$ is a fixed constant,  $B(t)$ tends to zero as $n$ tends to infinity, with high probability. 
\end{theorem}
Please refer to appendix for the proof.

\section{Empirical Evaluation}
\label{sec:empirical}

\begin{figure*}[t]
  \centering
\begin{minipage}[b]{\textwidth}
  \subfloat[$\alpha=0.65$, $n_u/n_s=1$]{
    \includegraphics[width=0.3\textwidth]{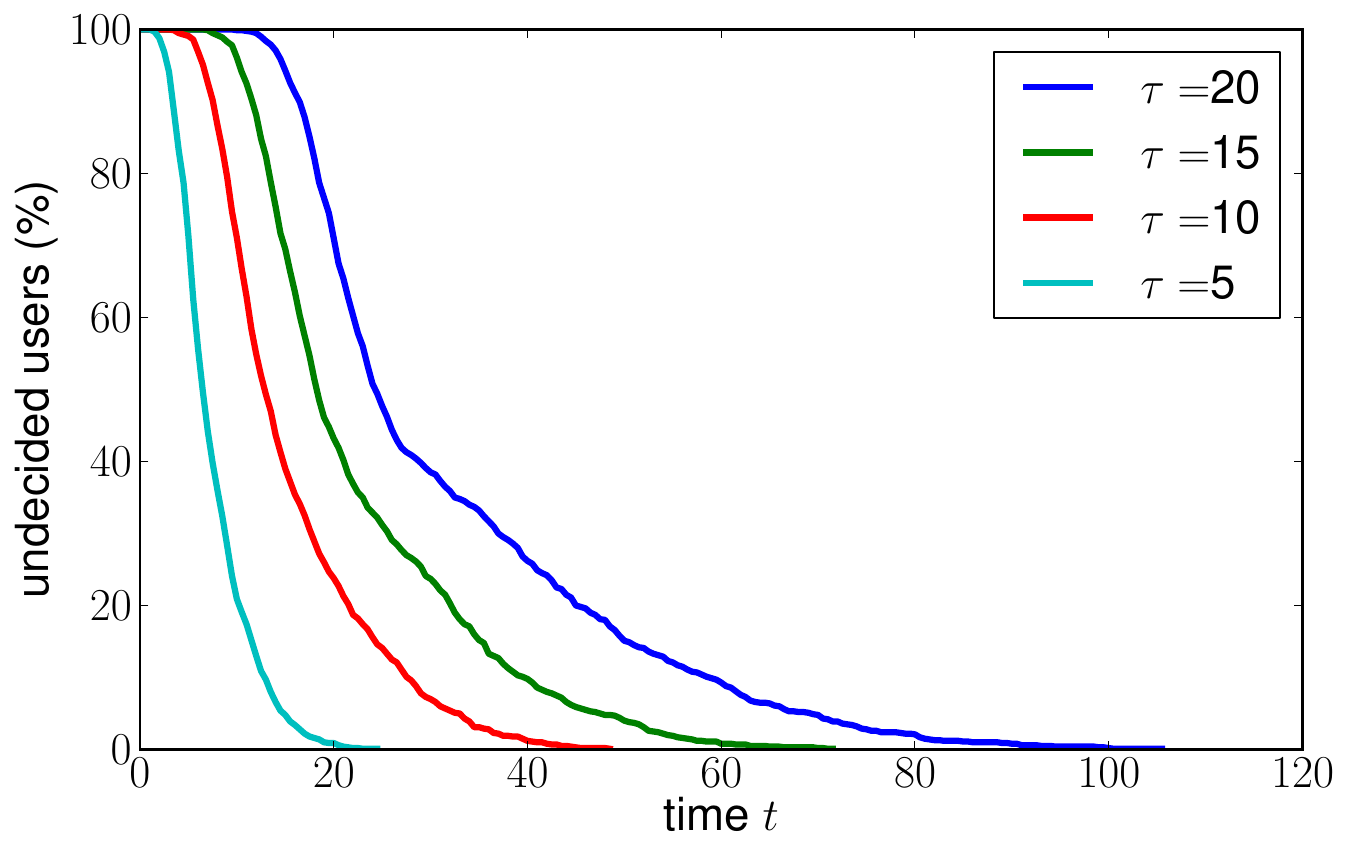}
  }
  \subfloat[$\alpha=0.65$, $n_u/n_s=10$]{
    \includegraphics[width=0.3\textwidth]{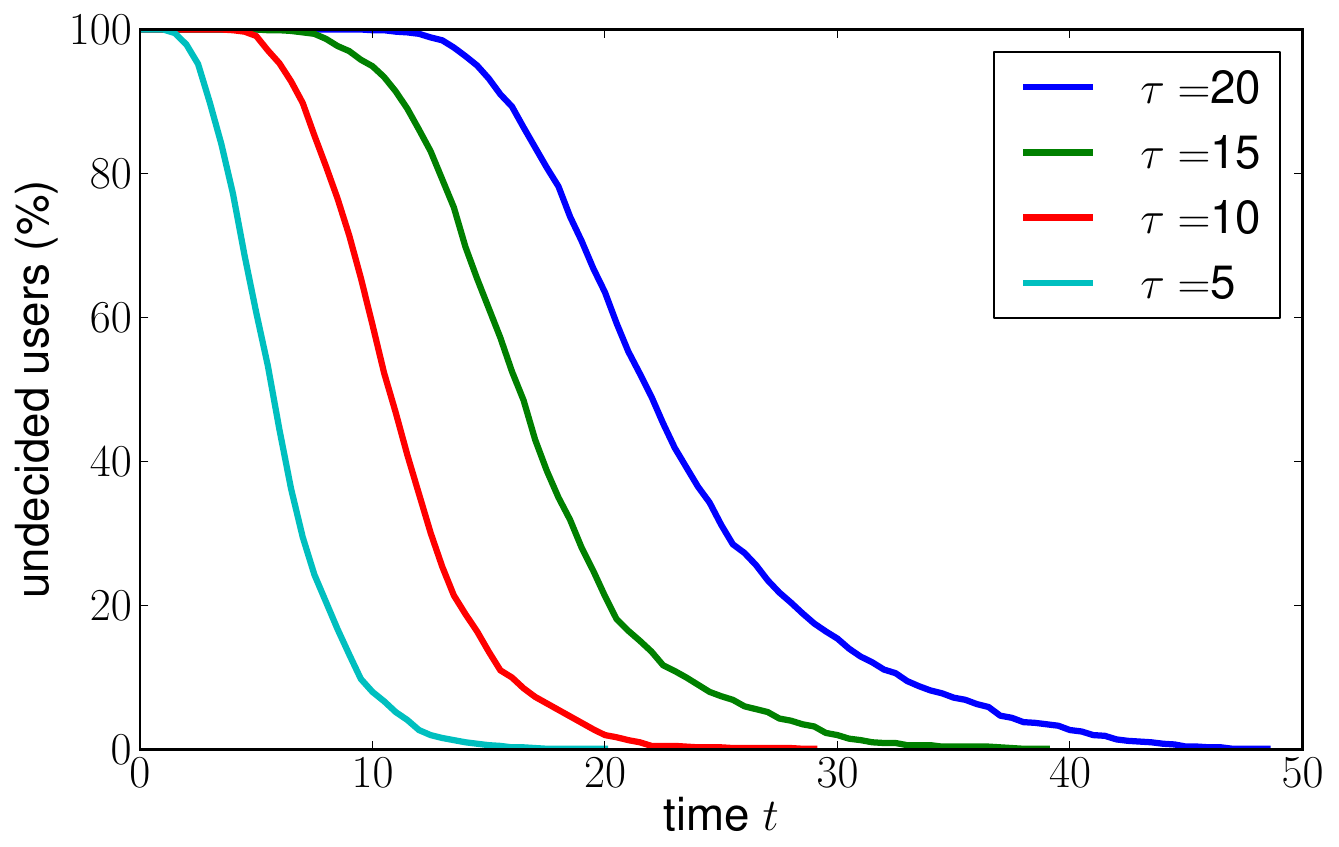}
  }
\subfloat[$\alpha=0.65$, $n_u/n_s=20$]{
   \includegraphics[width=0.3\textwidth]{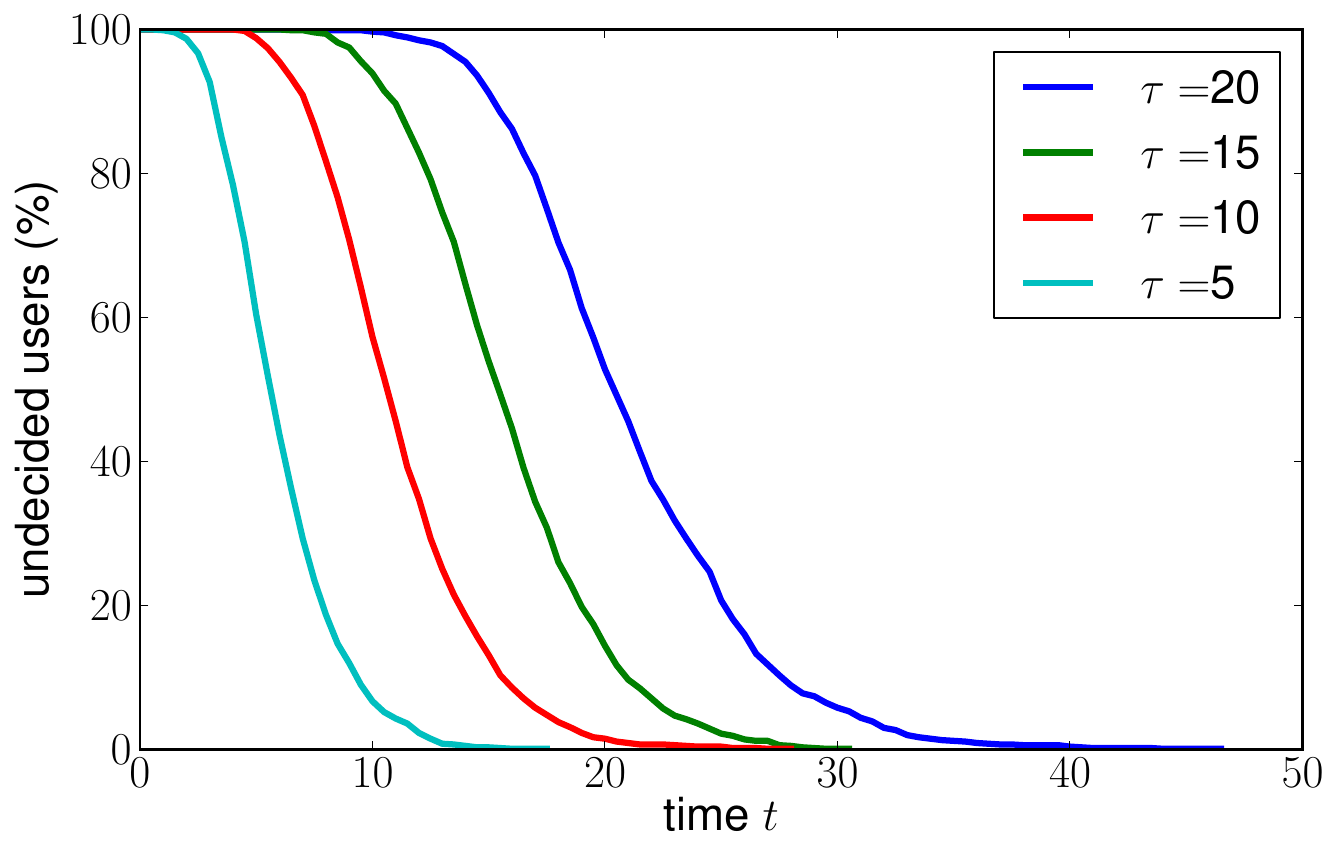}
 }
  \caption{The figures show the percentage of undecided users for a typical power law distribution ($\alpha=0.65$) with spread $\sigma = 2$ and $n_u = 1000$. Note that the undecided users decrease with time in all cases, but the convergence is faster when we use fewer but larger servers by setting $n_u/n_s$ to be larger.   Also, the smaller values of the look-ahead window $\tau$ result in faster convergence.}
\label{fig:assign}
\end{minipage}
\vspace{0.2in}
\begin{minipage}[b]{\textwidth}
  \centering
  \subfloat[$\alpha=0.6$, $n_u/n_s=1$]{
    \includegraphics[width=0.36\textwidth]{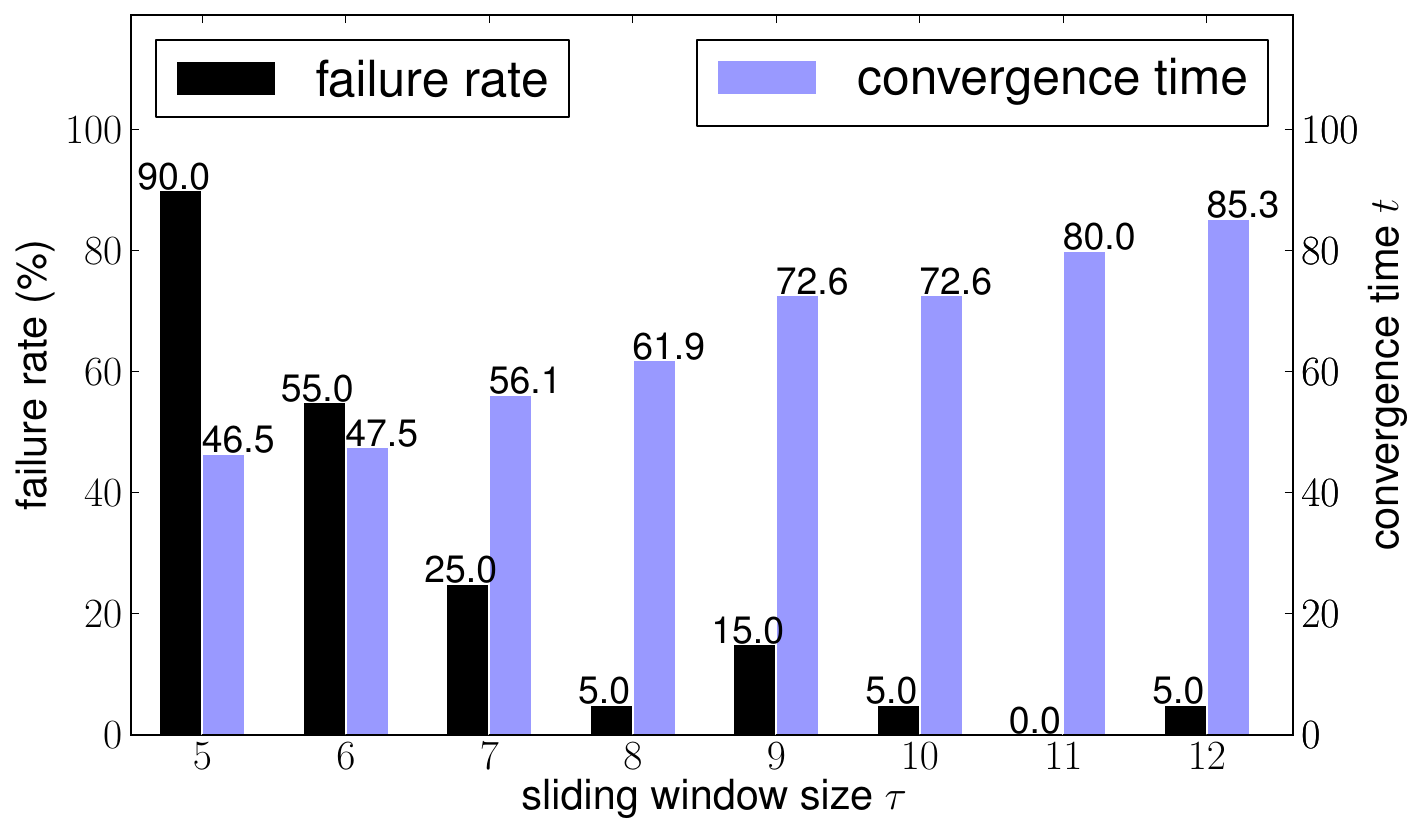}
  }
  \hspace{1.1in}
  \subfloat[$\alpha=0.6$, $n_u/n_s=20$]{
    \includegraphics[width=0.33\textwidth]{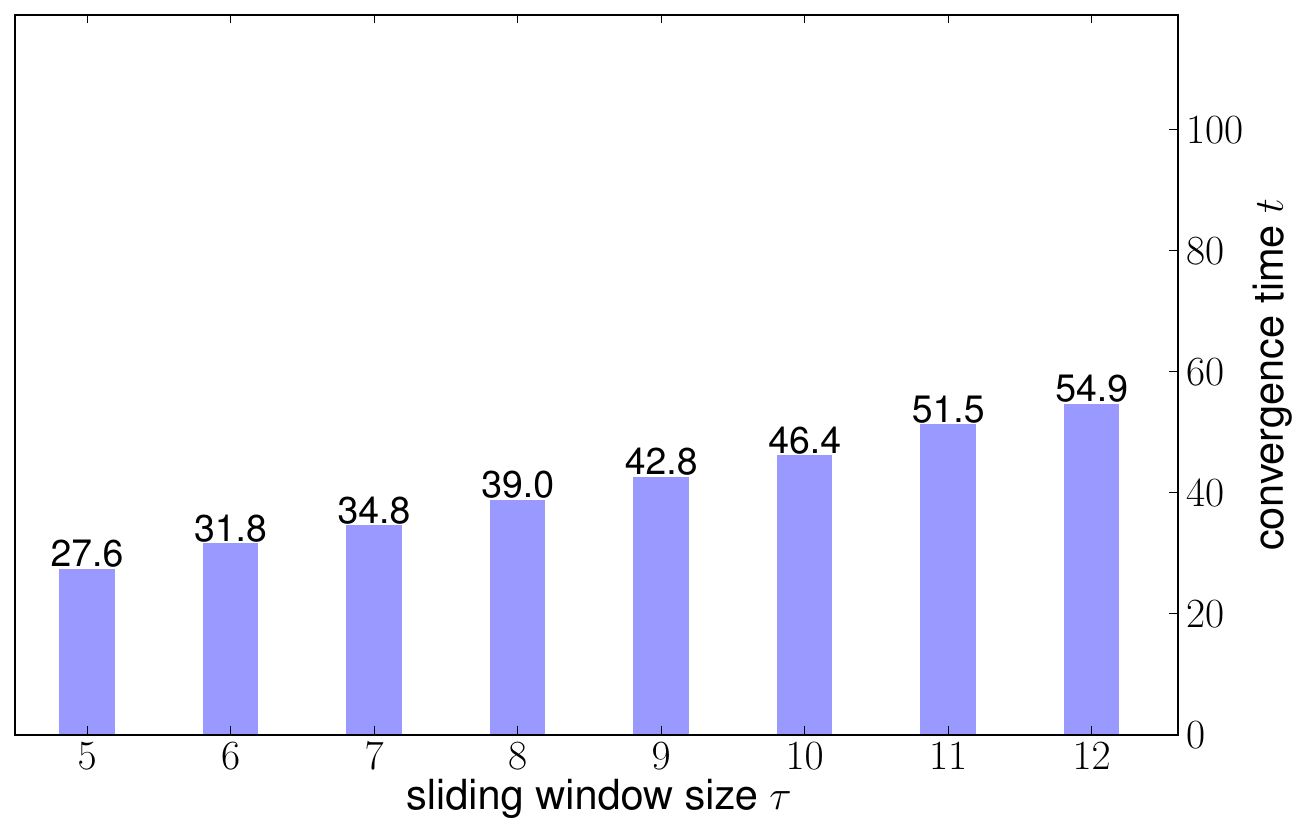}
 }
\caption{Generally, as $\tau$ increases, convergence time increases but  failure rate decreases. It is also true for larger servers ($n_u/n_s$ = 20), only the failure has gone to zero for all investigated sliding window size$\tau$.}
\label{fig:tau}
\end{minipage}

\end{figure*}

\begin{figure*}[t]
  \centering
\begin{minipage}[b]{0.45\textwidth}
\centering
\includegraphics[width=0.68\textwidth]{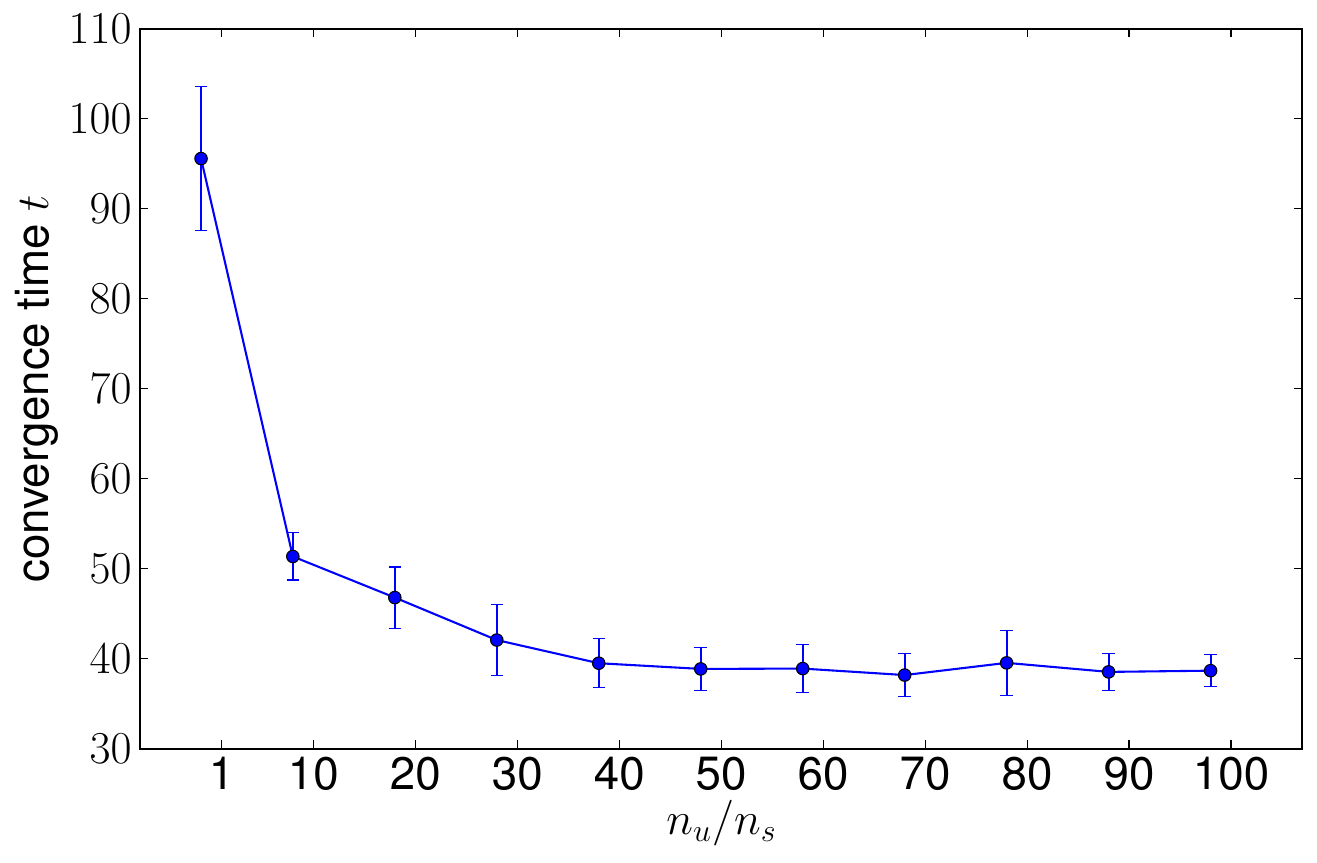}
\caption{As $n_u/n_s$ increases  fewer servers with larger capacity are used and convergence time decreases. The decrease is less pronounced beyond  $n_u/n_s\geq 40$ under this setting ($\alpha=0.65$, $\sigma=2$, $\tau=20$).}
\label{fig:utos}
\end{minipage}
\qquad
\begin{minipage}[b]{0.45\textwidth}
\centering
\includegraphics[width=0.68\textwidth]{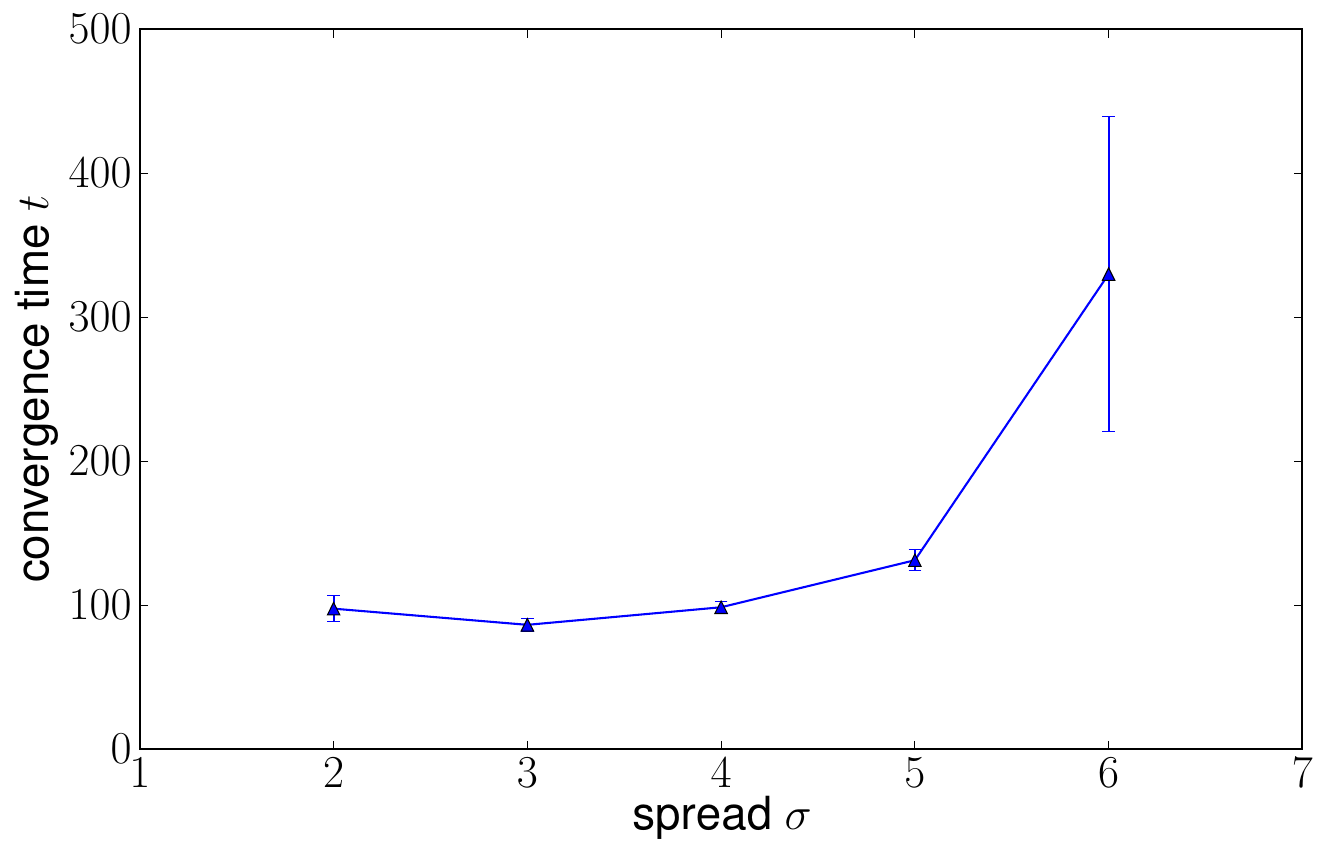}
\caption{There is a very small incremental benefit in using $\sigma = 3$ instead of  $2$, though higher values of $\sigma > 3$ only increased the convergence time. ($\alpha=0.65, n_u/n_s = 1, \tau=20,\kappa=2.$)}
\label{fig:sigma}
\end{minipage}
\end{figure*}

\begin{figure*}[t]
  \centering
\begin{minipage}[b]{0.45\textwidth}
\centering
\includegraphics[width=0.7\textwidth]{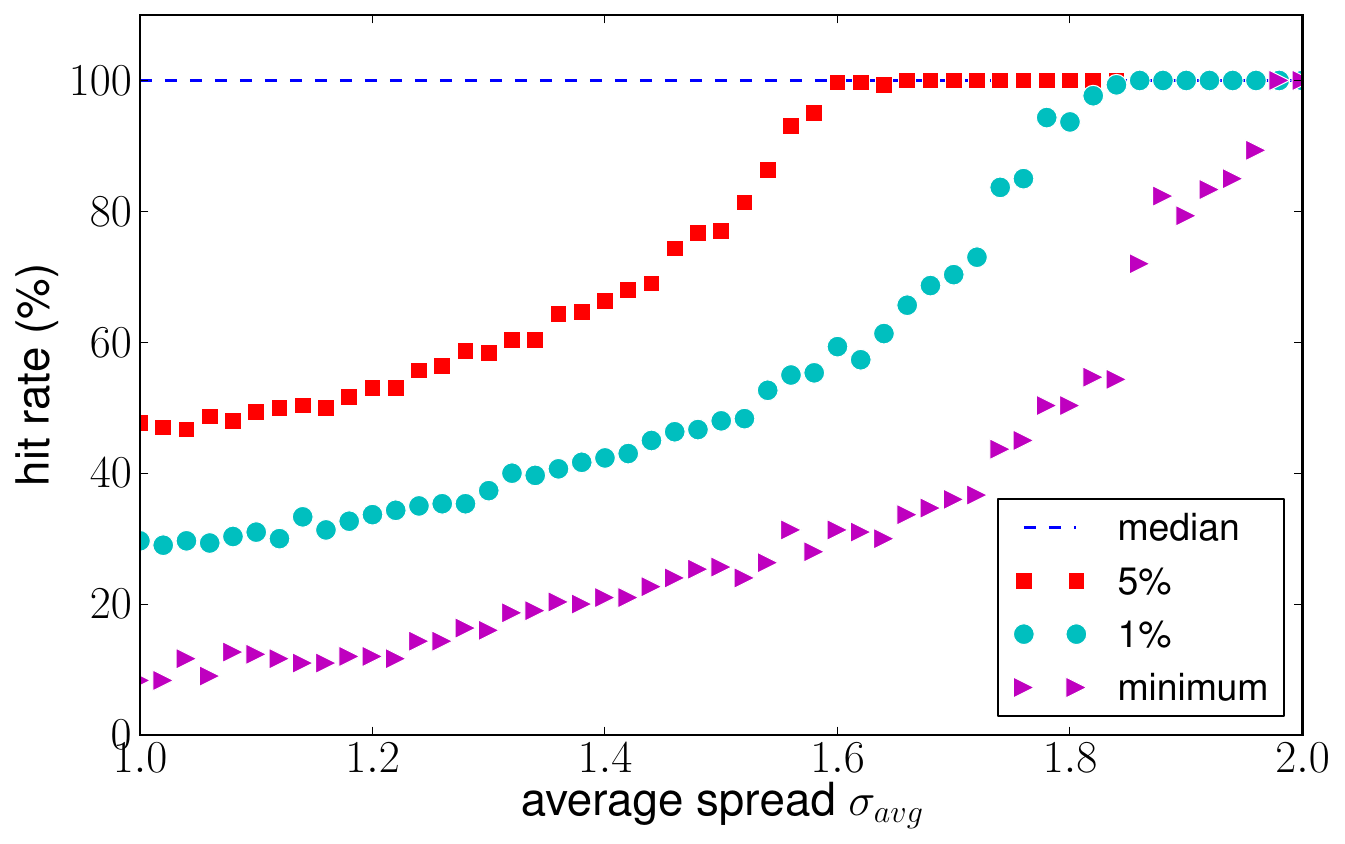}
\caption{Order statistics of the hit rate of the user population. ($\alpha=0.65, n_u/n_s = 1, \tau=10,\kappa=2.$)}
\label{fig:mix}
\end{minipage}
\qquad
\begin{minipage}[b]{0.45\textwidth}
\centering
\includegraphics[width=0.7\textwidth]{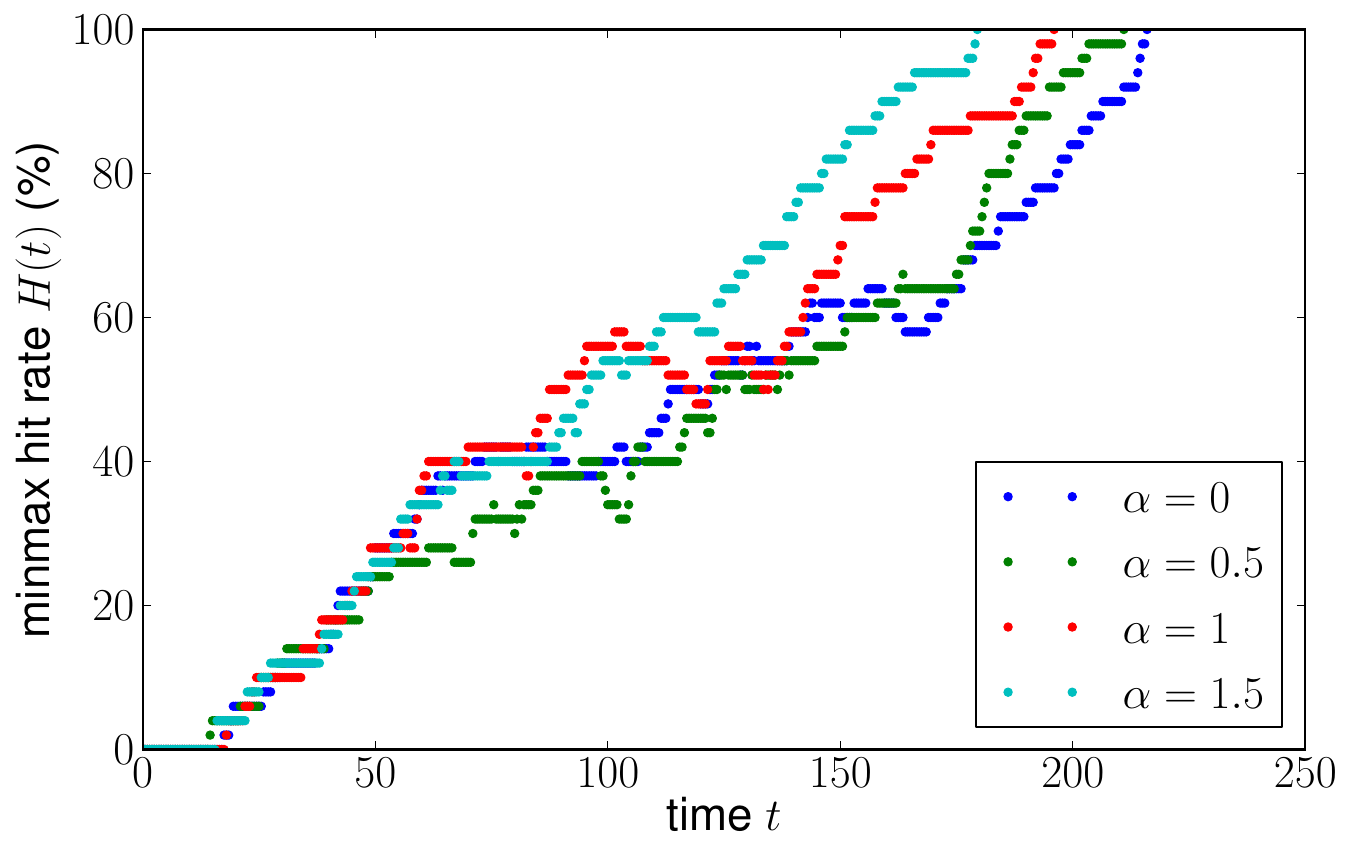}
\caption{Minmax hitrate versus time for different power law distributions.}
\label{fig:dist}
\end{minipage}
\end{figure*}

We empirically study our algorithm $\textsl{GoWithTheWinner}$ by building a simulator. Each user is implemented as a Poisson arrival sequence with unit rate.  We use $n_u = 1000$ users. To simulate varying numbers of servers, users, and applications, we also varied $n_s$  and $n_a$  such that $1 \leq n_u/n_c, n_u/n_s \leq 100$. We also simulate a range of values for the spread $1 \leq \sigma \leq 6$, and sliding window size $1 \leq \tau \leq 20$. Each server implements an LRU application replacement policy of size $\kappa \geq 2$. The applications are requested by users using the power law distribution of Equation~\ref{eq:powerlaw} with $\alpha = 0.65$ to model realistic content popularity \cite{Breslau:1999} \cite{Gill:2007}. However, we also vary $\alpha$ from $0$ (uniform distribution) to $1.5$ in some of our simulations.

The system is said to have {\em converged\/} when all users have decided on a single server from their set of candidate servers. There are two complementary metrics that relate to convergence.  {\em Failure rate} is the probability that the system converged to a non-optimal state where there exists servers that are overbooked, resulting in some users  incurring application misses after convergence occurred. The failure rate is measured by performing the simulation multiple times and assessing the goodness of the converged state. {\em Convergence time} is the time it takes for the system to converge provided that it converged to an optimal state.

\subsection{Speed of convergence}


Figure \ref{fig:assign} shows how the fraction of undecided users decreases over time till it reaches zero, resulting in convergence. Note that users do not decide in the first $\tau$ steps, since they must wait at least that long to accumulate a window of $\tau$ application hits. However, once the first $\tau$ steps complete,  the decrease in the number of undecided users is fast as users discover that at least one of their two randomly chosen candidate servers have less load. The rate of decrease in undecided users slows down again towards the end, as pockets of users who experience cache contention in {\em both} of their server choices require multiple iterations to resolve.

In this simulation, we keep the number of users $n_u = 1000$ but vary the number of servers $n_s$ to achieve different values for $n_u/n_s$. Note that for a fair comparison, we keep the system-wide load  the same.  Load $l$ is a measure of cache contention in the network and is naturally defined as the ratio of the numbers of users in the system and total serving capacity that is available in the system. That is, $l  \stackrel{\Delta}{=} n_u/ (\kappa \cdot n_s)$.
For all three setting of Figure~\ref{fig:assign}, we keep load $l = 0.5$. The figure shows that with fewer (but larger) servers ($n_u/n_s$ is larger) the convergence time is faster, because having server capacity in a few larger servers provides a larger application hit rate than having the same capacity in several smaller servers. Similar performance gains are also found in the context of web caching and parallel jobs scheduling \cite{Sparrow}. The convergence times are plotted explicitly in 
Figure~\ref{fig:utos} for a greater range of  user-to-server ratios. As $n_u/n_s$ increases from $1$ to $40$, convergence time decreases. The decrease in convergence times are not significant beyond $n_u/n_s \geq 40$.

\subsection{Impact of sliding window $\tau$}

%

The sliding window $\tau$ is the number of recent requests used by algorithm $\textsl{GoWithTheWinner}$ to estimate the hit rate. As shown in Figure ~\ref{fig:tau}, there is a natural tradeoff between convergence time and failure rate. When $\tau$ increases, the users take longer to converge, as they require a $100\%$ hit rate in a larger sliding window. However, waiting for a longer period also makes their decisions more robust. That is, a user is  less likely to choose an overbooked server, since an overbooked server is less likely to provide a string of $\tau$ application hits for large $\tau$. In our simulations with many smaller caches ($n_u/n_s = 1$), when $\tau \leq 4$, users made quick choices based on a smaller sliding window. But, this resulted in the system converging to a non-optimal state 100\% of the time. As $\tau$ further increases, the failure rate decreased. The value of $\tau = 11$ is a suitable sweet spot as it  results in the  smallest  convergence time for a zero failure rate. However, for fewer but larger servers ($n_s/n_u = 20$), all selections of window size $\tau$ (thus the small values like $\tau=5$) yielded a $0\%$ failure rate, while the convergence time still increases as the window size gets larger.


\subsection{Impact of spread $\sigma$}
As shown in Theorems~\ref{thm:hitrate} and~\ref{thm:onechoice}, a spread of $\sigma \geq 2$  is required for the system to converge to an optimal solution, while a spread of $\sigma=1$ is insufficient. As predicted by our analysis, our simulations did not converge to an optimal state with $\sigma = 1$. Figure \ref{fig:sigma} shows the  convergence time as a function of spread, for  $\sigma \geq 2$. 


As  $\sigma$ increases, there are two opposing factors that impact the convergence time. The first factor is that  as $\sigma$ increases, each user has more choices and the user is more likely to find a suitable server with less load. On the other hand, an increase in $\sigma$ also increases the total number of initial requests in the system that equals $\sigma n_u$. Thus, the system starts out in a state where servers have greater average load when $\sigma$ is larger. These opposing forces result in a very small incremental benefit when using $\sigma = 3$ instead of $2$, though the higher values of $\sigma > 3$ showed no benefit as convergence time increases with $\sigma$ increases.

We established the ``power of two random choices'' phenomenon where two or more random server choices yield superior results to having just one. It is intriguing to ask {\em what percentage of users need two choices} to reap the benefits of multiple choices? Consider a mix of users, some with two random choices and others with just one. Let  $\sigma_{avg}$, $1 \leq \sigma_{avg} \leq 2$, denote the average value of the spread among the users.

\comment{
\begin{figure}[ht]
\centering
\includegraphics[width=0.4\textwidth]{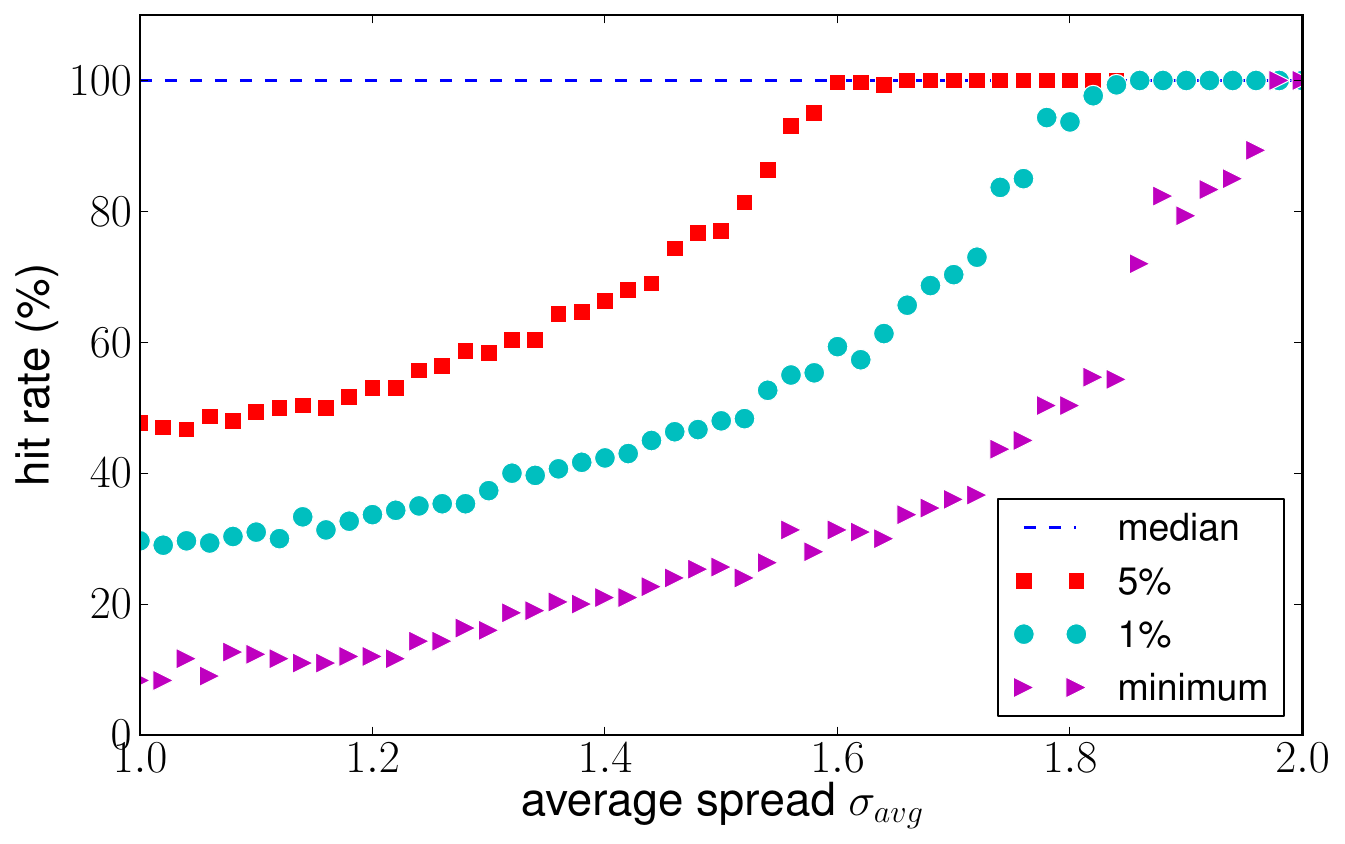}
\caption{Order statistics of the hit rate of the user population. ($\alpha=0.65, n_u/n_s = 1, \tau=10,\kappa=2.$)}
\label{fig:mix}
\end{figure}
}

In Figure \ref{fig:mix}, we show different order statistics of the  hit rate as a function of $\sigma_{avg}$. Specifically, we plot the minimum value, $1^{st}$-percentile, $5^{th}$- percentile and the median ($50^{th}$-percentile) of the hit rates of the users after running the system for a long enough period of 200 time units. As our theory predicts, when $\sigma_{avg} = 2$, the minimum and all the order statistics converge to $100\%$, as all users converge to a $100\%$ hit rate. Further, if we are interested in only the median user, any value of the spread is sufficient to guarantee that $50\%$ of the users obtain a $100\%$ hit rate. Perhaps the most interesting phenomena is that if $\sigma_{avg} = 1.7$, i.e., $70\%$ of the users have two choices and the rest have one choice, the $5^{th}$-percentile converges to $100\%$, i.e., all but $5\%$ of the users experience a $100\%$ hit rate. For a higher value of $\sigma_{avg} = 1.9$, the $1^{st}$-percentile converges to $100\%$, i.e., all but the $1\%$ of the users experience a $100\%$ hit rate. This result shows that our algorithm still provides benefits even if only {\em some} users  have multiple random choices of servers available to them.

\comment{
\subsection{Impact of load}
\begin{figure}[t]
\centering
\includegraphics[width=0.4\textwidth]{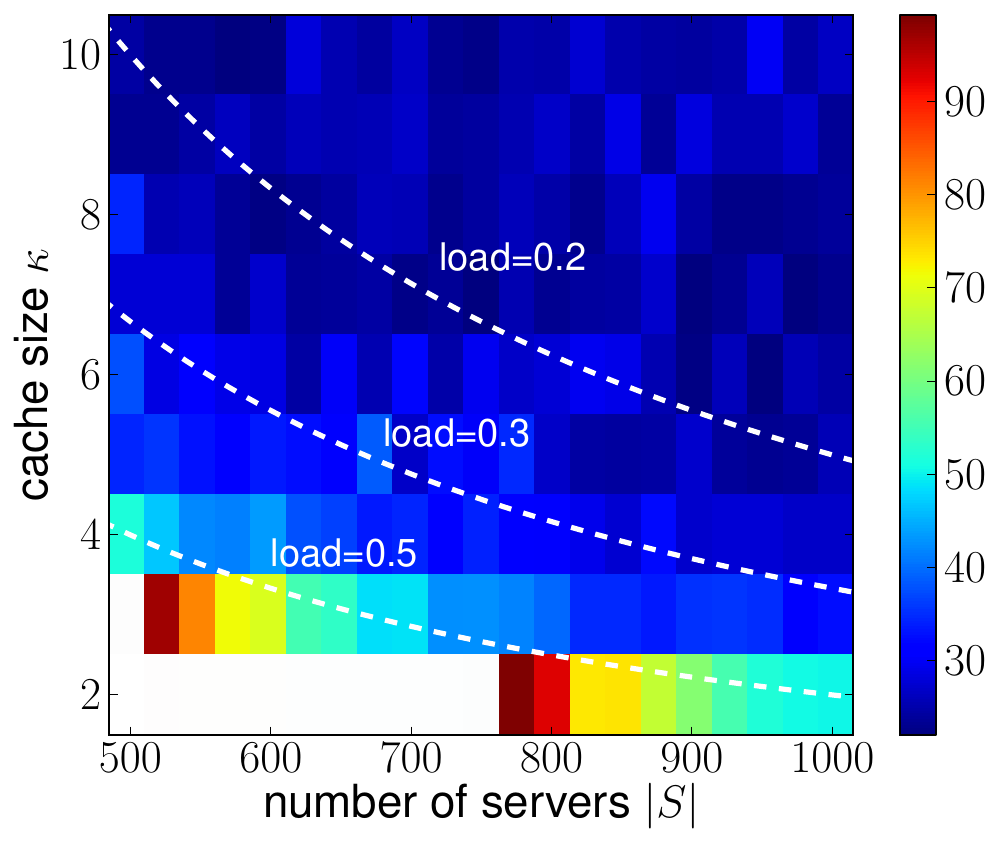}
\caption{Heatmap shows that convergence is faster if the cache size ($\kappa$) and/or number of servers ($n_s$) is larger for same number of users ($n_u$). ($\alpha=0.65, \tau=10,\sigma=4.$)}
\label{fig:heat}
\end{figure}
Now we look at the performance of our algorithm under varying load conditions. Recall load $l  = n_u/ (\kappa \cdot n_s)$.
Figure \ref{fig:heat} shows convergence time as a heatmap with different
values of cache size $\kappa$ and number of servers $n_s$. The convergence time varies from lower (blue) to higher (red) values. The figure also shows contours where the load is fixed. The white block in the bottom left corner represents 
high load closer to $1$ where our algorithm did not converge within a reasonable amount of time. It is easy to see that convergence time is faster when either the number of servers or the cache size increases. Figure~\ref{fig:load} shows this trend more explicitly as the convergence timße increases with increasing load.
The reason for this behavior is that as load increases there is more contention for the cache resources requiring more time for the algorithm to find a suitable server for each user.
}

\subsection{Impact of demand distribution}


We now study  how hit rate changes with the exponent $\alpha$ in the power law distribution of Equation~\ref{eq:powerlaw}.  Note that the distribution is uniform when $\alpha = 0$ and is the harmonic distribution when $\alpha = 1$.  As $\alpha$ increases, since the tails fall as a power of $\alpha$,  the distribution gets more and more skewed towards applications with a smaller rank. In  Figure \ref{fig:dist}, we plot the minmax hitrate over  time for different $\alpha$, where we see that a larger $\alpha$ leads to faster convergence. The reason is that as the popularity distribution gets more skewed, a larger fraction of users can share the same VMs for popular applications, leading to better hit rate and faster convergence. 
Thus, the uniform application popularity distribution ($\alpha=0$)  is the worst case and the algorithm converges faster for the distributions that tend to occur more commonly in practice. Providing theoretical support for this empirical result by analyzing the convergence time to show faster convergence for larger $\alpha$ is a topic for future work.

To summarize our results from empirical evaluation: We establish an inverse relationship between the length of the history used for performance evaluation (denoted by  $\tau$) and the failure rate defined as the probability that the system converges to a non-optimal state. We show that as $\tau$ increases the convergence time increases, but the failure rate decreases. We  also empirically evaluate the impact of the number of choices of candidate servers. We show that two or more random choices are required for all users to receive a $100\%$ application hit rate. Though even if only  70\% of the users  make two choices,  it is sufficient for $95\%$  of the users to receive a $100\%$ application hit rate.

\section{Related work}
Server selection algorithms have a rich history of both research and actual implementations over the past two decades. Several server selection algorithms have been proposed and empirically evaluated, including client-side algorithms that use historical performance feedback using probes \cite{dykes2000empirical,crovella1995dynamic}.  Server selection has also been studied in a variety of contexts, such as  the web \cite{crovella1995dynamic,sayal1998selection}, video streaming\cite{torres2011dissecting}, and cloud services\cite{wendell2010donar}. Our work is distinguished from the prior literature in that we theoretically model the ``Go-With-The-Winner'' paradigm that is common to many proposed and implemented client-side server selection algorithms. Our work is the first formal study of the efficacy and convergence of such algorithms.

In terms of analytical techniques, our work is closely related to  prior work on balls-into-bins games where the witness tree technique was first utilized \cite{mitzenmacherRS2001}. Witness trees were subsequently used to analyze load balancing algorithms, and circuit-switching algorithms \cite{cole1998randomized}. However, our setting involves additional complexity requiring novel analysis due to the fact that users can share a single cached copy of an object and the hit rate feedback is only a probabilistic indicator of server overload. Also,  our work shows that the ``power of two random choices'' phenomenon  applies in the context of content delivery, a phenomenon known to hold in other contexts such as  balls-into-bins, load balancing, relay allocation for services like Skype \cite{Nguyen:2008}, and circuit switching in interconnection networks \cite{mitzenmacherRS2001}. 
\comment{
There are also parallels between our work on bit rate maximization for video delivery and the recent work on throughput maximization in multi-path communication \cite{Peter:2007}. Similar in spirit to our work, the work on multi-path communication concludes that coordinated rate control over multiple randomly selected paths provides optimal throughput and is significant better than uncoordinated control where rates are determined independently over each single path. However, our work also evolves techniques for determining the convergence rate that is potentially applicable in these domains.
 }
 
 \section{Conclusion}
 Our work constitutes the first formal study of the simple ``Go-With-The-WInner'' paradigm in the context of web and video content delivery.  For web (resp., video) delivery, we proposed a simple algorithm where each user randomly chooses two or more candidate servers and selects the server that provided the best hit rate (resp., bit rate). We proved that the algorithm converges quickly to an optimal state where all users receive the best hit rate (resp., bit rate) and no server is overloaded, with high probability.  While we make some assumptions to simplify the theoretical analysis, our simulations evaluate a broader setting that incorporates a range of values for $\tau$ and $\sigma$, varying  content popularity distributions, differing load conditions, and situations where only some users have multiple server choices. Taken together, our work establishes that the simple ``Go-With-The-Winner'' paradigm can provide algorithms that converge quickly to an optimal solution, given a sufficient number of random choices and a sufficiently (but not perfectly) accurate performance feedback.
 
\bibliographystyle{abbrv}
\bibliography{ServerSelectionRefs,new_ref,GoWithTheWinner}
\balance
\newpage
\appendix

\subsection{Detailed Proof of Theorem \ref{thm:hitrate}}
\begin{proof}
For simplicity, we prove the situation where $\sigma = 2$, i.e., each user initially chooses two random candidate servers in step 1 of the algorithm. The case where $\sigma > 2$ is analogous. 
\comment{Further we will assume that a user $u$ requesting application from an overbooked server $s$  at time $t$,  $0 \leq t \leq T$, has $H_\tau(i,j,t) <100\%$,  as the probability that this assumption is violated at most $1/ n^{\Omega(1)}$ as per Lemma~\ref{lem:overbookhit}.}  Wlog, we also assume  $\kappa$ is at most $O(\log n/ \log\log n)$, which includes the interesting case of  $\kappa$ equal to a constant.  When the server capacity is larger, i.e., if  $\kappa = \Omega(\log n/ \log\log n)$,  there will be no overbooked servers with high probability and the theorem holds trivially. This observation follows from a well-known result that if $n$ balls (i.e., users)  uniformly and randomly select $k$  out of $n$ bins (i.e. servers), then the maximum number of users that select a server is $O(\log n / \log\log n)$ with high probability, when  $k$ is a fixed constant \cite{raab1998balls}. 

In contradiction to the theorem, suppose some user $u$ has not decided on a server  by time $T$. We construct a ``witness tree\footnote{A witness tree is so called as it bears witness to the occurrence of an event such as  a user being undecided.}''  of degree $\kappa + 1$ and depth at least $\rho$, where  $\rho = T/\delta =  \kappa \log\log n/\log (\kappa + 1)$. Each node of the witness tree is a server. Each edge of the witness tree is a user whose two nodes correspond to the two servers chosen by that user.  We show that the existence of an undecided user in time step $T$ is unlikely by enumerating all possible witness trees and showing that the occurrence of any such  witness tree is unlikely.  The proof proceeds in the following three steps.

\noindent{\bf (1) Constructing a witness tree.}  If algorithm $\mathrm{MaxHitRate}$ has not converged to the optimal state at time $T$, then there exists a user (say $u_1$)  and a server $s$ such that $H_\tau(u_1, s, T) < 100\%$, since user $u_1$ has not yet found a server with a $100\%$ hit rate. We make server $s$ the root of the witness tree.

We find children for the root $s$ to extend the witness tree as follows. Since $H_\tau(u_1, s, T) < 100\%$, by Lemma \ref{lem:timefortau} we know server $s$ is overbooked at time $t' = t - \delta $, i.e., there are at least $\kappa + 1$  users requesting server $s$ for $\kappa + 1$ distinct applications at time $t' $. Let $u_1,\ldots,u_{\kappa+1}$ be the users who sent requests to server $s$ at time $t' $. Wlog, assume that the users $\{u_i\}$ are ordered in ascending order of their IDs.  By Lemma~\ref{lem:overbookhit}, we know that the probability of a user deciding on an overbooked server is small, i.e., at most $1/n^{\Omega(1)}$. Thus, with high probability, users $u_1,\ldots,u_{\kappa+1}$  are undecided at time $t'$ since server $s$ is \textsl{overbooked}. Let $s_i$ be the other server choice associated with user $u_i$ (one of the choices is server $s$).  We extend the witness tree by creating $\kappa +1$ children for the root $s$, one corresponding to each server $s_i$.  Note that for each of the servers $s_i$ we know that $H(u_i, s_i, t' ) < 100\%$, since otherwise user $u_i$ would have decided on server $s_i$ in time step $t' $. Thus, analogous to how we found children for $s$, we can recursively find $\kappa + 1$ children for each of the servers $s_i$ and grow the witness tree to an additional level.   

Observe that to add an additional level of the witness tree we went from server $s$ at time $T $ to servers $s_i$ at time $t' $, i.e., we went back in time by an amount of $T - t' \leq \delta $. If we continue the same process, we can construct a witness tree that is a $(\kappa + 1)$-ary tree of depth $T/\delta = \rho$.

\noindent{\bf (2) Pruning the witness tree.}
If the nodes of the witness tree are guaranteed to represent distinct servers, proving our probabilistic bound is relatively easy. The reason is that if the servers are unique then the users that represent edges of the tree are unique as well. Therefore the probabilistic choices that each user makes is independent, making it easy to evaluate the probability of occurrence of the tree. However, it may not be the case that the servers in the witness tree constructed above are unique, leading to dependent choices that are hard to resolve. Thus, we create a {\em pruned witness tree} by removing repeated servers from the  original (unpruned) witness tree. 

We prune the witness tree by visiting the nodes of the witness tree iteratively in breadth-first  search order starting at the root. As we perform breadth-first search (BFS), we remove (i.e., prune) some nodes of the tree and the subtrees rooted  at these nodes. What is left after this process is the pruned witness tree. We start by visiting the root. In each iteration, we visit the next  node $v$ in BFS order that has not been pruned. Let $\beta(v)$ denote the nodes visited {\it before\/} $v$. If $v$ represents a server that is different from the servers represented by nodes in $\beta(v)$, we do nothing.
 Otherwise, prune all nodes in the subtree rooted at $v$. Then, mark the edge from $v$ to its parent  as a {\em pruning edge}. (Note that the pruning edges are not part of the pruned witness tree.) The procedure continues until either no more nodes remain to be visited or there are $\kappa + 1$ pruning edges. In the latter case, we apply a final pruning by removing all nodes that are yet to be visited, though this step does not produce any more pruning edges. This process results in a pruned witness and a set of $p$ (say) pruning edges.  

 Note that each pruning edge corresponds to a user who we will call a {\em pruned user}. We now make a pass through the pruning edges to select a set $P$ of unique pruned users. Initially, $P$ is set to $\emptyset$. We visit the pruning edges in BFS order and for each pruning edge $(u,v)$ we add the user corresponding to $(u,v)$ to $P$, if this user is distinct from all users currently in $P$ and if $|P|<  \lceil p /2 \rceil $, where $p$ is the total number of pruning edges. We stop adding pruned users to set $P$ when we have exactly $\lceil p/2 \rceil$  users. Note that since a user who made server choices of $u$ and $v$ can appear at most twice as a pruned edge, once with $u$ in the pruned witness tree and once with $v$ in the pruned witness tree. Thus, we are guaranteed to find $\lceil p/2 \rceil$ distinct pruned users. 
 
After the pruning process, we are left with a pruned witness tree with nodes representing distinct servers and edges representing distinct users. In addition, we have a set $P$ of $\lceil p/2 \rceil$ distinct pruned users, where $p$ is the number of pruning edges. 

\noindent{\bf (3)  Bounding the probability of pruned witness trees.} 
We enumerate possible witness trees and bound their probability using the union bound. Observe that since the (unpruned) witness tree is a $(\kappa + 1)$-ary tree of depth $\rho$, the number of nodes in the witness tree is 
\begin{equation}
m = \sum_{0 \leq i \leq \rho} (\kappa + 1)^i =  \frac{(\kappa+1)^{\rho + 1} - 1}{\kappa} \leq 2 \log^2 n, \label{eq:mbound}
\end{equation}
since $\rho = 2 \log\log n / \log (\kappa + 1)$ and hence $(\kappa + 1)^\rho = \log^2 n$.

{\em Ways of choosing the shape of the pruned witness tree}:
The shape of the pruned witness tree is determined by choosing the $p$ pruning edges of the tree. The number of ways of selecting the $p$ pruning edges is at most ${m  \choose p} \leq m^{p},$ since there are at most $m$ edges in the (unpruned) witness tree. \\
{\em Ways of choosing users and
servers for the nodes and edges of the pruned witness tree}: The enumeration
proceeds by considering the nodes in BFS order.  The number of ways of
choosing the server associated with the root is $n$. Consider the $i^{th}$ internal node $v_i$ of the pruned witness tree
whose server has already been chosen to be $s_i$. Let $v_i$ have $\mu_i$
children. There are at most ${ {n} \choose {\mu_i} }$ ways of
choosing distinct servers for each of the $\mu_i$ children of
$v_i$.  Also, since there are at most $n$ users in the system
at any point in time, the number of ways to choose 
distinct users for the
$\mu_i$ edges incident on $v_i$ is also at most
${ {n} \choose {\mu_i}}$. There are ${\mu_i}!$
ways of pairing the users and the servers.
Further, the probability that a chosen user
chooses server $s_i$ corresponding to node $v_i$ and a specific one of
$\mu_i$ servers chosen above for $v_i$'s children is $$
\frac{1}{{n \choose 2}} = \frac{2}{n (n - 1)},$$ since each set of two servers is equally likely to be chosen in step 1 of the algorithm. Further, note that each of the $\mu_i$ users chose $\mu_i$ distinct applications and let the probability of occurrence of this event be $Uniq(n_a, \mu_i)$. This uniqueness probability has been studied in the context of collision-resistant hashing and it is known \cite{bellare2004hash} that
$Uniq(n_a, \mu_i)$ is largest when the content popularity distribution is the uniform distribution ($\alpha = 0$) and progressively becomes smaller as $\alpha$ increases. In particular, 
$Uniq(n_a, \mu_i) \leq e^{-\Theta(\mu^2_i/n_a)} < 1.$
Putting it together, the number of ways of choosing a distinct
server for each of the $\mu_i$ children of $v_i$, choosing a distinct
user for each of the $\mu_i$ edges incident on $v_i$, choosing a distinct application for each user, and  multiplying by the appropriate probability is at most 
\begin{equation}
\label{eq:expbd}
 {n \choose \mu_i} \cdot {n \choose \mu_i} \cdot {\mu_i}! \cdot \left (
\frac{2}{n (n - 1)}\right)^{\mu_i} \cdot Uniq(n_a, \mu_i) \leq \frac{2^{\mu_i}}{
\mu_i!}, 
\end{equation} 
provided $\mu_i > 1$.
Let $m'$ be the number of internal nodes $v_i$ in the pruned witness tree
such that $\mu_i = \kappa + 1$. Using the bound in Equation~\ref{eq:expbd}
for only these $m'$ nodes, the number of ways of choosing the users
and servers for the nodes and edges respectively of the pruned witness tree
weighted by the probability that these choices occurred is at most $$n \cdot (2^{\kappa+1} / (\kappa+1)!)^{m'}.$$
{\em Ways of choosing the pruned users in $P$}: Recall that there are $\lceil p/2 \rceil$ distinct pruned users  in $P$. The number of ways of choosing the users in $P$ is at
most ${n^{\lceil p/2 \rceil}}$, since at any time step there are at most $n$ users in the system to choose from. Note that a pruned user has both of its
server choices in the pruned witness tree. Therefore, the probability that a given user is a pruned user is at most $ m^2/
n^2.$ Thus the number of choices for the
$\lceil p/2 \rceil$ pruned users in $P$ weighted by the probability
that these pruned users occurred is at most
$${n^{\lceil p/2\rceil}} \cdot (m^2 / n^2)^{\lceil p/2
\rceil}
\leq (m^2/ n)^{\lceil p/2 \rceil}.$$
{\em Bringing it all together}: The probability that
there exists a pruned
witness tree with $p$ pruning edges, and $m'$ internal nodes with $(\kappa+1)$ children each, is at most
\begin{eqnarray}
m^ {p} \cdot  n \cdot (2^{\kappa+1} / (\kappa+1)!)^{m'} \cdot  (m^2/ n)^{\lceil p/2 \rceil} \nonumber \\
\leq n \cdot
(2^{\kappa+1} / (\kappa+1)!)^{m'}  \cdot (m^4 /n)^{\lceil p/2 \rceil} \nonumber  \\
\leq n \cdot (2e / (\kappa+1))^{m' (\kappa + 1)}  \cdot (m^4 /n)^{\lceil p/2 \rceil}, 
 \label{eq:bd2}
\end{eqnarray}
since $(\kappa + 1)! \geq ((\kappa + 1)/e)^{\kappa + 1}$.
There are two possible cases depending on how the pruning process terminates. If the number of pruning edges, $p$, equals $\kappa + 1$ then the third term of Equation~\ref{eq:bd2} is 
$$(m^4 /n)^{\lceil p/2 \rceil} \leq (16\log^8 n /n)^{\lceil (\kappa + 1)/2 \rceil} \leq 1/ n^{\Omega(1)},$$
using Equation~\ref{eq:mbound} and assuming that cache size $\kappa$ is at least  a suitably large constant.
Alternately, if the pruning process terminates with fewer than $\kappa + 1$ pruning edges, it must be that at least one of the $\kappa + 1$ subtrees rooted at the children of the root $s$ of the 
(unpruned) witness tree has no pruning edge. Thus, the number of internal nodes $m'$ of the pruned witness tree with $(\kappa + 1)$ children  each is bounded as follows:
$$m' = \sum_{0 \leq i < \rho -1 } (\kappa + 1)^i \geq (\kappa + 1)^{\rho - 2} \geq \log^2 n / (\kappa + 1)^2,$$ 
as $(\kappa + 1)^\rho = \log^2 n$. Thus, the second term of Equation~\ref{eq:bd2} is 
$$(2e / (\kappa+1))^{m' (\kappa + 1)}\leq (2e / (\kappa+1))^{\log^2 n/ (\kappa + 1)} \leq 1/n^{\Omega(1)},$$ 
assuming $\kappa > 2e -1$  but is at most $O(\log n/\log\log n)$.
Thus, in either case, the bound in
Equation~\ref{eq:bd2} is $1/n^{\Omega(1)}$. Further, since there
are at most
$m$ values for $p$, the total probability of a pruned
witness tree is at most
$m \cdot 1/n^{\Omega(1)}$ which is
$1/n^{\Omega(1)}$. This completes the proof of the theorem. 
\end{proof}


\subsection{Proof of Theorem~\ref{thm:onechoice}}

\begin{proof}
From the classical analysis of throwing $n$ balls into $n$ bins \cite{mitzenmacherRS2001}, we know that there exist a subset $U' \subseteq U$ such that  $|U'| = \Theta(\log n /\log\log n)$ and all users in $U'$ have chosen a single server $s$, with high probability. Now we show that some user in $U'$ must have a small hit rate with high probability. Let $C'$ represent the set of all objects accessed by all users in $S'$. The probability that $|C'| \leq \kappa w(n)$ can be upper bounded as follows, where $w(n)$ is an arbitrarily slowly growing function of $n$. The number of ways of picking $C'$ objects from a set $C$ of $n$ objects is at most $n^{|C'|}$. The probability that a user  in $U'$ will pick an object in $C'$ can be upper bounded by the probability that a user chooses one of the $|C'|$ most popular objects. Thus the probability that a user in $U'$ picks an object in $C'$  is at most $\mathcal{H}(|C'|,\alpha)/\mathcal{H}(n,\alpha) = \Theta( (|C'| / n)^{1-\alpha})$, where $\mathcal{H}(i,\alpha)$ is the $i^{th}$ generalized harmonic number and $\mathcal{H}(i,\alpha)=\Theta(i^{1 -\alpha})$.Thus, the probability that all users in $U'$ pick objects in $C'$  is at most $\Theta((|C'| /n)^{(1 - \alpha)|U'|})$.  Therefore, the probability that $|C'| \leq \kappa w(n)$ is at most
\begin{eqnarray}
& n^{|C'|}  \cdot \Theta((|C'| /n)^{(1 - \alpha)|U'|}) \nonumber\\
 \leq & n^{ \kappa w(n)} \cdot (\kappa w(n)/n)^{\Theta(\log n/\log\log n)}  = o(1) \nonumber
\end{eqnarray}
Thus, probability that $|C'| \leq \kappa w(n)$ is small and hence $|C'| > \kappa w(n)$, with high probability. Since the minmax hit rate $H(t)$ is at most $\kappa/|C'|$ which is at most $1/w(n)$, $H(t)$ tends to zero with high probability.
\end{proof}


\subsection{Proof of Theorem \ref{lem:nu>ns}}
\begin{proof}
We prove the lemma using Chernoff Bound. \\
Recall that there are $n_s$ servers and $n_u$ users. Let $X_{ij}$ be the binary indicator of user $j$ selects server $i$. Because users selects servers uniformly at random, $X_{ij}\sim \text{Bernoulli}(1/n_s)$. Thus, $Y_i=\sum_{j=1}^{n_u}X_{ij}, \E{Y_i} = \frac{n_u}{n_s}$.
\begin{enumerate}
\item [(1)] When $n_s=n, n_u=n\log n$, we have $\E{Y_i} = \log n$. For any $\delta>0$, we have for the maximum load over servers $Y$,
\begin{align*}
\prob{Y>(1+\delta)\log n} &\leq \sum_{i=1}^{n_s} \prob{Y_i>(1+\delta)\log n} \\
& = n \prob{Y_i>(1+\delta)\log n} \\
& \leq n e^{-\log n \delta^2/3} \\
& = n^{1-\delta^2/3},
\end{align*}
which equals to $n^{-\Theta(1)}$ if $\delta>\sqrt{3}$. Thus, with high probability, $Y\leq (1+\sqrt{3})\log n$. And if the server capacity $\kappa\geq (1+\sqrt{3})\log n$, all users will have hit rate $1$.
\item [(2)] When $n_s=n, n_u=n^{\alpha}, \alpha>1$, following the same argument as in (1), for any $\delta>0$,
\begin{align*}
\prob{Y>(1+\delta)n^{\alpha-1}} &\leq \sum_{i=1}^{n_s} \prob{Y_i>(1+\delta)n^{\alpha-1}} \\
& = n \prob{Y_i>(1+\delta)n^{\alpha-1}} \\
& \leq n e^{-\log n \frac{n^{\alpha-1}\delta^2}{3\log n}}\\
& = n^{1-\frac{n^{\alpha-1}\delta^2}{3\log n}},
\end{align*}
which equals to $n^{-\omega(1)}$. Thus the maximum load on all servers is $Y=n^{\alpha-1}=\frac{n_u}{n_s}$. And as long as the capacity of servers $\kappa\geq n^{\alpha-1}=\frac{n_u}{n_s}$, all users will get hit rate $1$.
\end{enumerate}
\end{proof}


\subsection{Proof for Theorem \ref{thm:bitrate}}

The proof is similar to that of Theorem~\ref{thm:hitrate} in that we create a witness tree, prune  it, and then show that a pruned witness tree is unlikely. However, Algorithm MaxBitRate differs with Algorithm GoWithTheWinner differs in that it's a synchronous algorithm that all users make requests in synchronization and the algorithm executes in discrete time steps rather than continuous time scale. Thus before the proof, we need the following lemmas to assist the formal proof.

\begin{lemma}
\label{lem:cachemiss}
For any time $t > 0$, if user $u$ receives a application miss from server $s$ at some time $t$,  then server $s$ is overbooked at time $t - 1$.
\end{lemma}
\begin{proof}
If user $u$ requested a service $c_u$ from server $s$ at time $t$, it must have also requested $c_u$ from server $s$ at time $t - 1$.  As soon as the request  for $c_u$ was processed at time $t -1$, it was placed server $s$. There must have been $\kappa$ other requests for distinct services that caused the service replacement policy to evict $c_u$,  resulting in the application miss at time $t$. Thus, at least $\kappa + 1$ distinct services were requested from server $s$ at time $t - 1$, i.e., server $s$ is overbooked at time $t -1$.
\end{proof}

To prove convergence, we choose the {\em sliding window size} $\tau = c \kappa \log n$, for a suitably  large positive constant $c$. Further, consider an initial time interval from time zero to time $T$ that consists of  $\rho$ intervals of length $\tau + 1$ each, where $\rho = 2 \log\log n/ \log (\kappa + 1)$.  Thus,  $T = \rho \cdot (\tau + 1) =  O(\kappa \log n \log\log n/\log (\kappa + 1))$. 
\begin{lemma}\label{lem:overbookhit2}
The probability that some user $u \in U$ decides on an overbooked server $s \in S$ at some time $t$, $0 \leq t \leq T$,  is at most  $1/n^{\Omega(1)}$. 
\end{lemma}
\begin{proof}
Suppose user $u$ decides on an overbooked server $s$ at time $t$. Then, it must be the case that 
$H_\tau(u,s,t) = 100\%$. Thus, server $s$ provided a application hit to user $u$ in every time $t'$, $t - \tau <  t' \leq  t$. Recall that each server serves simultaneous requests by first batching the requests according to the requested applications, i.e., each batch contains requests for the same application, and then serving each batch in random order. Since server $s$ is overbooked at time $t$, it must have been overbooked during all the previous time steps. An overbooked server $s$ has at least $\kappa + 1$ distinct applications being requested, i.e., it has at least $\kappa + 1$ batches of requests. The request made by user $u$ will receive a application miss if the batch in which it belongs to is $\kappa + 1$ or higher in the random ordering. Thus, the probability that user $u$ receives a application miss from the overbooked server $s$ in any time step $t ' \leq t$ is at least $1 / (\kappa + 1)$. Since $H_\tau(u,s,t)$ is $100\%$ only if there is no application miss at any time $t'$, $t - \tau<  t' \leq  t$, the probability of such an occurrence is at most
$$
\left( 1- \frac{1}{\kappa+1}\right)^\tau = \left(1- \frac{1}{\kappa +1}\right)^{c \kappa \log n} = \frac{1}{n^{\Omega(1)}},
$$
since $\tau = c \kappa \log n$.  Using the union bound and choosing a suitably large constant $c$, the probability that there exists a user $u \in U$ who decides on an overbooked server $s$ at some $t$, $0 \leq t \leq T$,  is at most
$$ n \times n  \times (T + 1) \times \frac{1}{n^{\Omega(1)}} \leq \frac{1}{n^{\Omega(1)}},$$
since there are $n$ users, at most $n$ overbooked servers, and $T + 1 = O(\kappa \log n \log\log n/\log (\kappa+1))$ time steps. 
\end{proof}

Now with the two lemmas above, we can prove Theorem~\ref{thm:bitrate} as the following.

\begin{proof}
For simplicity, we prove the situation where $\sigma = 2$, i.e., each user initially chooses two random candidate servers in step 1 of the algorithm. The case where $\sigma > 2$ is analogous. 
\comment{Further we will assume that a user $u$ requesting application from an overbooked server $s$  at time $t$,  $0 \leq t \leq T$, has $H_\tau(i,j,t) <100\%$,  as the probability that this assumption is violated at most $1/ n^{\Omega(1)}$ as per Lemma~\ref{lem:overbookhit}.}  Wlog, we also assume  $\kappa$ is at most $O(\log n/ \log\log n)$, which includes the practically interesting case of  $\kappa$ equal to a constant.  When the server capacity is larger, i.e., if  $\kappa = \Omega(\log n/ \log\log n)$,  there will be no overbooked servers with high probability and the theorem holds trivially. This observation follows from a well-known result that if $n$ balls (i.e., users)  uniformly and randomly select $k$  out of $n$ bins (i.e. servers), then the maximum number of users that select a server is $O(\log n / \log\log n)$ with high probability, when  $k$ is a fixed constant \cite{raab1998balls}. 

In contradiction to the theorem, suppose some user $u$ (say) has not decided on a server  by time $T$. We construct a ``witness tree\footnote{A witness tree is so called as it bears witness to the occurrence of an event such as  a user being undecided.}''  of degree $\kappa + 1$ and depth at least $\rho$, where  $\rho = T/(\tau+1) =  2 \log\log n/\log (\kappa + 1)$. Each node of the witness tree is a server. Each edge of the witness tree is a user whose two nodes correspond to the two servers chosen by that user.  We show that the existence of an undecided user in time step $T$ is unlikely by enumerating all possible witness trees and showing that the occurrence of any such  witness tree is unlikely.  The proof proceeds in the following three steps.

\noindent{\bf (1) Constructing a witness tree.}  If algorithm $\mathrm{MaxHitRate}$ has not converged to the optimal state at time $T$, then there  exists an user (say $u_1$)  and a server $s$ such that $H_\tau(u_1, s, T-1) < 100\%$, since user $u_1$ has not yet found a server with a $100\%$ hit rate. We make server $s$ the root of the witness tree.

We find children for the root $s$ to extend the witness tree as follows. Since $H_\tau(u_1, s, T-1) < 100\%$, there exists a time $t'$, $T  - 1 - \tau < t' \leq T -1$, such that user  $u_1$ received a application miss from server $s$.  By Lemma~\ref{lem:cachemiss},  server $s$ was overbooked at time $t' - 1$, i.e., there are at least $\kappa + 1$  users requesting server $s$ for $\kappa + 1$ distinct applications at time $t' - 1$. Let $u_1,\ldots,u_{\kappa+1}$ be the users who sent requests to server $s$ at time $t' - 1$. Wlog, assume that the users $\{u_i\}$ are ordered in 
ascending order of their IDs.  By Lemma~\ref{lem:overbookhit}, we know that the probability of a user deciding on an overbooked server is small, i.e., at most $1/n^{\Omega(1)}$. Thus, with high probability, users $u_1,\ldots,u_{\kappa+1}$  are undecided at time $t' - 1$ since they made a request to an overbooked server $s$. Let $s_i$ be the other server choice associated with user $u_i$ (one of the choices is server $s$).  We extend the witness tree by creating $\kappa +1$ children for the root $s$, one corresponding to each server $s_i$.  Note that for each of the servers $s_i$ we know that $H(u_i, s_i, t' -2) < 100\%$, since otherwise user $u_i$ would have decided on server $s_i$ in time step $t' -2$. Thus, analogous to how we found children for $s$, we can recursively find $\kappa + 1$ children for each of the servers $s_i$ and grow the witness tree to an additional level.   

Observe that to add an additional level of the witness tree we went from server $s$ at time $T -1$ to servers $s_i$ at time $t' -2$, i.e., we went back in time by an amount of $T - 1 - (t' - 2) \leq \tau + 1$. If we continue the same process, we can construct a witness tree that is a $(\kappa + 1)$-ary tree of depth $T/(\tau + 1) = \rho$.

\noindent{\bf (2) Pruning the witness tree.}
If the nodes of the witness tree are guaranteed to represent distinct servers, proving our probabilistic bound is relatively easy. The reason is that if the servers are unique then the users that represent edges of the tree are unique as well. Therefore the probabilistic choices that each user makes is independent, making it easy to evaluate the probability of occurrence of the tree. However, it may not be the case that the servers in the witness tree constructed above are unique, leading to dependent choices that are hard to resolve. Thus, we create a {\em pruned witness tree} by removing repeated servers from the  original (unpruned) witness tree. 

We prune the witness tree by visiting the nodes of the witness tree iteratively in breadth-first  search order starting at the root. As we perform breadth-first search (BFS), we remove (i.e., prune) some nodes of the tree and the subtrees rooted  at these nodes. What is left after this process is the pruned witness tree. We start by visiting the root. In each iteration, we visit the next  node $v$ in BFS order that has not been pruned. Let $\beta(v)$ denote the nodes visited {\it before\/} $v$. If $v$ represents a server that is different from the servers represented by nodes in $\beta(v)$, we do nothing.
 Otherwise, prune all nodes in the subtree rooted at $v$. Then, mark the edge from $v$ to its parent  as a {\em pruning edge}. (Note that the pruning edges are not part of the pruned witness tree.) The procedure continues until either no more nodes remain to be visited or there are $\kappa + 1$ pruning edges. In the latter case, we apply a final pruning by removing all nodes that are yet to be visited, though this step does not produce any more pruning edges. This process results in a pruned witness and a set of $p$ (say) pruning edges.  

 Note that each pruning edge corresponds to a user who we will call a {\em pruned user}. We now make a pass through the pruning edges to select a set $P$ of unique pruned users. Initially, $P$ is set to $\emptyset$. We visit the pruning edges in BFS order and for each pruning edge $(u,v)$ we add the user corresponding to $(u,v)$ to $P$, if this user is distinct from all users currently in $P$ and if $|P|<  \lceil p /2 \rceil $, where $p$ is the total number of pruning edges. We stop adding pruned users to set $P$ when we have exactly $\lceil p/2 \rceil$  users. Note that since a user who made server choices of $u$ and $v$ can appear at most twice as a pruned edge, once with $u$ in the pruned witness tree and once with $v$ in the pruned witness tree. Thus, we are guaranteed to find $\lceil p/2 \rceil$ distinct pruned users. 
 
After the pruning process, we are left with a pruned witness tree with nodes representing distinct servers and edges representing distinct users. In addition, we have a set $P$ of $\lceil p/2 \rceil$ distinct pruned users, where $p$ is the number of pruning edges. 

\noindent{\bf (3)  Bounding the probability of pruned witness trees.} 
We enumerate possible witness trees and bound their probability using the union bound. Observe that since the (unpruned) witness tree is a $(\kappa + 1)$-ary tree of depth $\rho$, the number of nodes in the witness tree is 
\begin{equation}
m = \sum_{0 \leq i \leq \rho} (\kappa + 1)^i =  \frac{(\kappa+1)^{\rho + 1} - 1}{\kappa} \leq 2 \log^2 n, \label{eq:mbound}
\end{equation}
since $\rho = 2 \log\log n / \log (\kappa + 1)$ and hence $(\kappa + 1)^\rho = \log^2 n$.

{\em Ways of choosing the shape of the pruned witness tree.}
The shape of the pruned witness tree is determined by choosing the $p$ pruning edges of the tree. The number of ways of selecting the $p$ pruning edges is at most ${m  \choose p} \leq m^{p},$ since there are at most $m$ edges in the (unpruned) witness tree. 

{\em Ways of choosing users and
servers for the nodes and edges of the pruned witness tree}. The enumeration
proceeds by considering the nodes in BFS order.  The number of ways of
choosing the server associated with the root is $n$. Consider the $i^{th}$ internal node $v_i$ of the pruned witness tree
whose server has already been chosen to be $s_i$. Let $v_i$ have $\delta_i$
children. There are at most ${ {n} \choose {\delta_i} }$ ways of
choosing distinct servers for each of the $\delta_i$ children of
$v_i$.  Also, since there are at most $n$ users in the system
at any point in time, the number of ways to choose 
distinct users for the
$\delta_i$ edges incident on $v_i$ is also at most
${ {n} \choose {\delta_i}}$. There are ${\delta_i}!$
ways of pairing the users and the servers.
Further, the probability that a chosen user
chooses server $s_i$ corresponding to node $v_i$ and a specific one of
$\delta_i$ servers chosen above for $v_i$'s children is $$
\frac{1}{{n \choose 2}} = \frac{2}{n (n - 1)},$$ since each set of two servers is equally likely to be chosen in step 1 of the algorithm. Further, note that each of the $\delta_i$ users chose $\delta_i$ distinct applications and let the probability of occurrence of this event be $Uniq(n_c, \delta_i)$. This uniqueness probability has been studied in the context of collision-resistant hashing and it is known \cite{bellare2004hash} that
$Uniq(n_c, \delta_i)$ is largest when the content popularity distribution is the uniform distribution ($\alpha = 0$) and progressively becomes smaller as $\alpha$ increases. In particular, 
$Uniq(n_c, \delta_i) \leq e^{-\Theta(\delta^2_i/n_c)} < 1.$
Putting it together, the number of ways of choosing a distinct
server for each of the $\delta_i$ children of $v_i$, choosing a distinct
user for each of the $\delta_i$ edges incident on $v_i$, choosing a distinct application for each user, and  multiplying by the appropriate probability is at most 
\begin{equation}
\label{eq:expbd}
 {n \choose \delta_i} \cdot {n \choose \delta_i} \cdot {\delta_i}! \cdot \left (
\frac{2}{n (n - 1)}\right)^{\delta_i} \cdot Uniq(n_c, \delta_i) \leq \frac{2^{\delta_i}}{
\delta_i!}, 
\end{equation} 
provided $\delta_i > 1$.
Let $m'$ be the number of internal nodes $v_i$ in the pruned witness tree
such that $\delta_i = \kappa + 1$. Using the bound in Equation~\ref{eq:expbd}
for only these $m'$ nodes, the number of ways of choosing the users
and servers for the nodes and edges respectively of the pruned witness tree
weighted by the probability that these choices occurred is at most $$n \cdot (2^{\kappa+1} / (\kappa+1)!)^{m'}.$$
{\em Ways of choosing the pruned users in $P$.} Recall that there are $\lceil p/2 \rceil$ distinct pruned users  in $P$. The number of ways of choosing the users in $P$ is at
most ${n^{\lceil p/2 \rceil}}$, since at any time step there are at most $n$ users in the system to choose from. Note that a pruned user has both of its
server choices in the pruned witness tree. Therefore, the probability that a given user is a pruned user is at most $ m^2/
n^2.$ Thus the number of choices for the
$\lceil p/2 \rceil$ pruned users in $P$ weighted by the probability
that these pruned users occurred is at most
$${n^{\lceil p/2\rceil}} \cdot (m^2 / n^2)^{\lceil p/2
\rceil}
\leq (m^2/ n)^{\lceil p/2 \rceil}.$$
{\em Bringing it all together.} The probability that
there exists a pruned
witness tree with $p$ pruning edges, and $m'$ internal nodes with $(\kappa+1)$ children each, is at most
\begin{eqnarray}
m^ {p} \cdot  n \cdot (2^{\kappa+1} / (\kappa+1)!)^{m'} \cdot  (m^2/ n)^{\lceil p/2 \rceil} \nonumber \\
\leq n \cdot
(2^{\kappa+1} / (\kappa+1)!)^{m'}  \cdot (m^4 /n)^{\lceil p/2 \rceil} \nonumber  \\
\leq n \cdot (2e / (\kappa+1))^{m' (\kappa + 1)}  \cdot (m^4 /n)^{\lceil p/2 \rceil}, 
 \label{eq:bd2}
\end{eqnarray}
since $(\kappa + 1)! \geq ((\kappa + 1)/e)^{\kappa + 1}$.
There are two possible cases depending on how the pruning process terminates. If the number of pruning edges, $p$, equals $\kappa + 1$ then the third term of Equation~\ref{eq:bd2} is 
$$(m^4 /n)^{\lceil p/2 \rceil} \leq (16\log^8 n /n)^{\lceil (\kappa + 1)/2 \rceil} \leq 1/ n^{\Omega(1)},$$
using Equation~\ref{eq:mbound} and assuming that cache size $\kappa$ is at least  a suitably large constant.
Alternately, if the pruning process terminates with fewer than $\kappa + 1$ pruning edges, it must be that at least one of the $\kappa + 1$ subtrees rooted at the children of the root $s$ of the 
(unpruned) witness tree has no pruning edge. Thus, the number of internal nodes $m'$ of the pruned witness tree with $(\kappa + 1)$ children  each is bounded as follows:
$$m' = \sum_{0 \leq i < \rho -1 } (\kappa + 1)^i \geq (\kappa + 1)^{\rho - 2} \geq \log^2 n / (\kappa + 1)^2,$$ 
as $(\kappa + 1)^\rho = \log^2 n$. Thus, the second term of Equation~\ref{eq:bd2} is 
$$(2e / (\kappa+1))^{m' (\kappa + 1)}\leq (2e / (\kappa+1))^{\log^2 n/ (\kappa + 1)} \leq 1/n^{\Omega(1)},$$ 
assuming $\kappa > 2e -1$  but is at most $O(\log n/\log\log n)$.
Thus, in either case, the bound in
Equation~\ref{eq:bd2} is $1/n^{\Omega(1)}$. Further, since there
are at most
$m$ values for $p$, the total probability of a pruned
witness tree is at most
$m \cdot 1/n^{\Omega(1)}$ which is
$1/n^{\Omega(1)}$. This completes the proof of the theorem. 
\end{proof}

\end{document}